\documentclass[11pt]{article}
\usepackage{amsmath,amssymb,amsfonts,amsthm,epsfig}
\usepackage[usenames,dvipsnames]{xcolor}
\usepackage{bm,xspace}
\usepackage{tcolorbox}
\usepackage{cancel}
\usepackage{fullpage}
\usepackage{liyang}
\usepackage{framed}
\usepackage{verbatim}
\usepackage{enumitem}
\usepackage{array}
\usepackage{multirow}
\usepackage{afterpage}
\usepackage{mathrsfs}
\usepackage{pifont} 
\usepackage{chngpage}
\usepackage[normalem]{ulem}
\usepackage{boxedminipage}
\usepackage{caption}

\newcommand{\lnote}[1]{\footnote{{\bf \color{blue}Li-Yang}: {#1}}}

\newcommand{\BayesOpt}{\textsc{BayesOpt}}

\usepackage{thm-restate}

\def\colorful{0}

\ifnum\colorful=1
\newcommand{\violet}[1]{{\color{violet}{#1}}}

\fi
\ifnum\colorful=0
\newcommand{\violet}[1]{{{#1}}}

\fi

\newcommand{\DDT}{\mathrm{DDT}}
\newcommand{\RDT}{\mathrm{RDT}}

\newcommand{\BPTIME}{\mathsf{BPTIME}}
\newcommand{\TIME}{\mathsf{TIME}}
\newcommand{\quasipoly}{\mathrm{quasipoly}}
\newcommand{\qbudget}{q_\mathrm{budget}}

\newtheorem*{Nisan}{Nisan's Theorem}

\newcommand{\pparagraph}[1]{\bigskip \noindent {\bf {#1}}}

\begin{document}

\title{Constructive derandomization of query algorithms\vspace*{10pt}}

\author{Guy Blanc \and \hspace{10pt} Jane Lange \vspace{8pt} \\
\hspace{6pt} {\small {\sl Stanford University}}
\and Li-Yang Tan}  

\date{\small{\today}}

\maketitle


\newtheorem*{notation*}{Notation}

\begin{abstract}
We give efficient deterministic algorithms for converting randomized query algorithms into deterministic ones.   We first give an algorithm that takes as input a randomized $q$-query algorithm~$R$ with description length~$N$ and a parameter $\eps$, runs in time $\poly(N) \cdot 2^{O(q/\eps)}$, and returns a deterministic $O(q/\eps)$-query algorithm $D$ that $\eps$-approximates the acceptance probabilities of~$R$.   These parameters are near-optimal: runtime $N + 2^{\Omega(q/\eps)}$ and query complexity $\Omega(q/\eps)$ are necessary.   

Next, we give algorithms for {\sl instance-optimal} and {\sl online} versions of the problem: 

\begin{itemize}[leftmargin=0.8cm]
\item[$\circ$] {\sl Instance optimal}: Construct a deterministic $q^\star_R$-query algorithm $D$, where $q^\star_R$ is minimum query complexity of any deterministic algorithm that $\eps$-approximates $R$.
\item[$\circ$] {\sl Online}: Deterministically approximate the acceptance probability of $R$ for a specific input~$\underline{x}$ in time $\poly(N,q,1/\eps)$, without constructing $D$ in its entirety.
\end{itemize}

%
%

Applying the techniques we develop for these extensions, we constructivize classic results that relate the deterministic, randomized, and quantum query complexities of boolean functions~(Nisan,~{\footnotesize{STOC 1989}}; Beals et al.,~{\footnotesize{FOCS 1998}}).  This has direct implications for the Turing machine model of computation: sublinear-time algorithms for total decision problems can be efficiently derandomized and dequantized with a subexponential-time preprocessing step. 


\end{abstract}

\thispagestyle{empty}

\newpage
\setcounter{page}{1}

\section{Introduction} 

The {\sl query model} is one of the simplest models of computation.  Each query to a coordinate of the input corresponds to one unit of computation, and the computational cost associated with an input is the number of its coordinates queried.  All other computation is considered free.

The query model is fundamental to both algorithms and complexity theory.  In algorithms, it is central to the study of {\sl sublinear-time computation}.   Since sublinear-time algorithms cannot afford to read the entire input, the number of input coordinates queried naturally becomes an important metric.  Indeed, there is a large body of work focused just on understanding the query complexity of algorithmic tasks across a broad range of areas spanning testing, optimization, and approximation (see e.g.~\cite{Rub06,CS10,Gol17} and the references therein).  The query model is also an important framework for the design and analysis of {\sl quantum algorithms}.  Many of the best known quantum algorithms, such as Grover's search~\cite{Gro96} and Shor's factoring algorithm~\cite{Sho99}, are captured by the quantum query model (see e.g.~\cite{Amb18} and the references therein).

In complexity theory, the query model is a model within which significant progress has been made on understanding of the overarching questions of the field.  A partial listing of examples include: the relationships between deterministic, randomized,  and nondeterministic computation (see e.g.~\cite{BdW02,Juk12}); the power and limitations of parallelism~\cite{CDR86,RVW18}; the complexity of search problems~\cite{LLMW95}; computing with noisy information~\cite{FRPU94}; direct sum~\cite{JKS10} and direct product theorems~\cite{NRS94,Sha04,Dru12}; etc.  In addition to being a fruitful testbed for developing intuition and techniques to reason about computation,  there is also a long history in complexity theory where results in the query model have been successfully bootstrapped to shed new light on much more powerful models such as communication protocols~\cite{RW99,GPW17,GPW18}, circuits and proof systems~\cite{GGKS18,dRMNPRV19}, and even Turing machines~\cite{FSS81,IN88,Ver99,Zim07,Sha11}. 

\subsection{This work: Constructive derandomization of query algorithms}  We study derandomization within the query model: the task of converting randomized query algorithms into deterministic ones.  The unifying focus of our work is on {\sl constructive} derandomization: rather than just establishing the {\sl existence} of a corresponding deterministic algorithm, our goal is to design {\sl efficient} meta-algorithms for constructing this deterministic algorithm.  In addition to being an aspect of derandomization that is natural and of independent interest, constructivity is also the key criterion that connects derandomization in the query model of computation (a non-uniform model) to derandomization in the Turing machine model of computation (a uniform model).  Constructive derandomization of query algorithms, and its implications for the Turing machine computation, have been previously studied by Impagliazzo and Naor~\cite{IN88}, Zimand~\cite{Zim07}, and Shaltiel~\cite{Sha11}; we give a detailed comparison of our work to prior work in~\Cref{sec:our-work}.  

\label{two-strands}
There are two main strands to this work.  First, we consider general randomized query algorithms~$R$, where we make no assumptions about the distribution of $R$'s output values on any given input $x$ (across possible outcomes of its internal randomness)---in particular, this distribution is not assumed to be concentrated on a certain value.  Here our goal is to deterministically approximate, for a given input $x$, the expected output value of $R$ when run on input $x$: 
\newpage
\begin{quote}
{\sl Given as input a randomized $q$-query algorithm $R : \zo^n \times \zo^m \to [0,1]$,  construct a deterministic $q'$-query algorithm $D : \zo^n \to [0,1]$ satisfying 
\begin{equation} \Ex_{\text{input $\bx$}}\Big[\big(D(\bx) - \Ex_{\text{randomness $\br$}} [R(\bx, \br)]\big)^2\Big] \le \eps. \label{eq:eps-approximating} 
\end{equation} 
We refer to $D$ as an $\eps$-approximating deterministic algorithm for $R$.}\footnote{All of our results can be stated more generally for algorithms with arbitrary real-valued output values; however, it will be convenient for us to assume a normalization where the output values are scaled to be in $[0,1]$.  Relatedly, note that if $R$ is $\zo$-valued, then $\Ex_{\br}[R(x,\br)] = \Prx_{\br}[R(x,\br) = 1]$ is simply the acceptance probability of $R$ on input~$x$.} 
\end{quote}
By Markov's inequality, (\ref{eq:eps-approximating}) implies that $|D(x) - \Ex_{\br}[R(x,\br)]| \le \eps$ for all but a $\sqrt{\eps}$-fraction of~$x$'s.  It is natural to seek a stronger {\sl worst-case} guarantee that holds for all $x$, but as we will show (and as is easy to see), there are simple examples of $q$-query randomized $R$'s for which any deterministic $D$ satisfying $|D(x) - \Ex_{\br}[R(x,\br)]| \le 0.1$ {\sl for all} $x$  has to have query complexity $q'$ where $q'$ is exponentially larger, or even unboundedly larger, than $q$. Therefore, without any added assumptions about $R$, any derandomization that does not incur such a blowup in query complexity has to allow for an average-case approximation such as~(\ref{eq:eps-approximating}). 

That brings us to the second strand of our work, where we focus on the special case of randomized query algorithms that compute boolean functions $f : \zo^n \to \zo$ with bounded error (or equivalently, randomized algorithms that solve total decision problems with bounded error).  These are randomized algorithms $R : \zo^n \times \zo^m \to \zo$ that are {\sl promised} to satisfy: 
\begin{equation} 
\label{eq:total} 
\text{For all $x \in \zo^n$}, \quad \Ex_{\br}[R(x,\br)] \in 
\begin{cases}
[\frac{2}{3}, 1] & \text{if $f(x) = 1$} \\
[0,\frac{1}{3}] & \text{if $f(x) = 0$.} 
\end{cases}
 \end{equation}
 Under such a promise, the aforementioned impossibility result ruling out a worst-case guarantee does not apply.  
 Indeed, in this case our goal will be that of achieving a {\sl zero-error} derandomization: to construct a deterministic query algorithm $D$ that {\sl computes $f$ exactly}, meaning that $D(x) = f(x)$ for all $x \in \zo^n$.

  
 \pparagraph{Efficiency of derandomization: the two criteria we focus on.}  
 \label{two-criteria} 
 In both settings---whether we are considering general randomized query algorithms, or those that solve total decision problems---we will focus on the two most basic criteria for evaluating the quality of a derandomization:  \begin{enumerate}
 \item[(i)] the runtime of the derandomization procedure; and 
 \item[(ii)] the query complexity of the resulting deterministic algorithm.  
 \end{enumerate} 
   That is, we seek a derandomization that is efficient in two senses: we would like to construct the corresponding deterministic query algorithm $D$ quickly, and we would like $D$'s query complexity to be as close to $R$'s query complexity as possible.

    \pparagraph{Perspectives from learning theory: random forests and latent variable models.}  For an alternative perspective on the objects and problems that we study in this work, in~\Cref{ap:learn} we discuss the roles that  randomized query algorithms play in the field of {\sl learning theory}, and the corresponding interpretations of the problem of constructive derandomization.

\subsection{Background: Non-constructive derandomization of query algorithms} 
\label{sec:background}

We begin by discussing two well-known results giving {\sl non-constructive} derandomizations of query algorithms, where the first of the two efficiency criteria discussed above, the runtime of the derandomization procedure, is disregarded.  These results establish the {\sl existence} of a corresponding deterministic query algorithm, but their proofs do not yield efficient algorithms for constructing such a deterministic  algorithm.  Looking ahead, the main contribution of our work, described in detail in~\Cref{sec:our-work}, is in obtaining constructive versions of these results.  

\begin{itemize}[leftmargin=0.5cm]
\item[$\circ$] In~\Cref{sec:Yao} we recall Yao's lemma~\cite{Yao77}, specializing it to the context of query algorithms.  For any randomized $q$-query algorithm $R$,  (the ``easy direction" of) Yao's lemma along with a standard empirical estimation analysis implies the existence of a deterministic $O(q/\eps)$-query algorithm that $\eps$-approximates $R$.

\item[$\circ$]In~\Cref{sec:Nisan} we recall Nisan's  theorem~\cite{Nis89}, which relates the deterministic and randomized query complexities of {\sl total decision problems}.  For every total decision problem $f$ that can be computed by a bounded-error randomized $q$-query algorithm, Nisan's theorem establishes the existence of a deterministic $O(q^3)$-query algorithm that computes $f$ exactly. 
\end{itemize}

These results are incomparable, and their proofs are very different: the first is essentially a simple averaging argument, whereas Nisan's theorem involves reasoning about the ``block sensitivity" of~$f$ and related boolean function complexity measures.  However, the two proofs share one common feature: they are both non-constructive. 


\subsubsection{The easy direction of Yao's  lemma} 
\label{sec:Yao} 

Yao's lemma~\cite{Yao77}, a special case of von Neumann's minimax theorem, is a simple and extremely useful technique in the study of randomized algorithms.  It shows that the bounded-error randomized complexity of a function $f : \zo^n \to \zo$ is an upper bound on its {\sl distributional complexity} relative to any distribution $\mu$ over $\zo^n$: the complexity of the optimal deterministic algorithm for $f$ that is correct on {\sl most} inputs, weighted according to~$\mu$.\footnote{This is in fact the ``easy direction" of Yao's lemma; the hard direction shows that the randomized complexity of $f$ is precisely equal to its distributional complexity relative to the worst distribution~$\mu$.}


Although this easy direction of Yao's lemma is most often applied in the context of randomized algorithms for decision problems, by combining its simple proof with a standard empirical estimation argument, one easily gets an extension to general randomized algorithms $R : \zo^n \times\zo^m \to [0,1]$, where no assumptions are made about the distribution of $R(x,\br)$.  We defer the proof of the following fact to~\Cref{ap:yao}.  


\begin{fact}[Non-constructive derandomization via the easy direction of Yao's lemma] 
\label{fact:yao} 
Let $R : \zo^n \times \zo^m \to [0,1]$ be a randomized $q$-query algorithm.  For every $\eps \in (0,\frac1{2})$, there exists a deterministic $O(q/\eps)$-query algorithm $D : \zo^n \to [0,1]$ satisfying
\[  \Ex_{\text{input $\bx$}}\Big[\big(D(\bx) - \Ex_{\text{randomness $\br$}} [R(\bx, \br)]\big)^2\Big] \le \eps. \] 
\end{fact}

We make two observations regarding the optimality of~\Cref{fact:yao}, the proofs of which are also deferred to~\Cref{ap:yao}: 


\begin{fact}[Optimality of query complexity]
\label{fact:lower-bound}
For every $q\in \N$ and $\eps \le O(q/n)$, there is a randomized $q$-query algorithm $R$ such that any $\eps$-approximating deterministic algorithm $D$ for~$R$ has to have query complexity $\Omega(q/\eps)$. 
\end{fact}

\begin{fact}[Impossibility of pointwise approximation]
\label{fact:pointwise-impossible} 
Consider the randomized $1$-query algorithm $R$ which on input $x$, samples $\bi \in [n]$ uniformly at random and outputs $x_{\bi}$.  Any deterministic algorithm $D$ satisfying $| D(x) - \Ex_{\br}[R(x,\br)]| \le 0.1$ for all $x \in \zo^n$ 
has to have query complexity $\Omega(n)$. 
\end{fact}
The example in~\Cref{fact:pointwise-impossible} is chosen to illustrate the largest possible gap ($1$ versus $\Omega(n)$).  Another canonical example is that of approximating the fractional Hamming weight of the input, for which the gap is $O(1)$ versus $\Omega(n)$.

\begin{remark}[Quantum analogue of~\Cref{fact:yao} and the work of Aaronson and Ambainis~\cite{AA14}]
\label{remark:AA} 
A major open problem in quantum complexity theory is that of obtaining a {\sl quantum} analogue of~\Cref{fact:yao}: showing\violet{---even just non-constructively---}that the acceptance probabilities of a quantum query algorithm $Q$ can be approximated on most inputs by a deterministic query algorithm (whose query complexity is polynomially related to that of $Q$'s).  For a precise formulation, see Conjecture 4 of~\cite{AA14}, where it is attributed as folklore dating back to 1999 or before. (See also~\cite{Aar05,Aar10,Aar08}.)

For one of our results (\Cref{thm:online}), we build on and extend techniques that Aaronson and Ambainis~\cite{AA14} developed to study this problem.
\end{remark} 
\subsubsection{Nisan's theorem} 
\label{sec:Nisan}

For the special case of randomized query algorithms that solve {\sl total decision problems} (recall (\ref{eq:total})), the impossibility result of~\Cref{fact:pointwise-impossible} does not apply.  Indeed, a classic result of Nisan~\cite{Nis89} establishes the existence of a {\sl zero-error} derandomization of such algorithms.  Given a function $f : \zo^n \to \zo$, we write $D(f)$ to denote its deterministic query complexity, and $R(f)$ to denote its bounded-error randomized query complexity.  (Please see~\Cref{sec:prelim} for formal definitions.)
\begin{Nisan}
\hypertarget{nisan-anchor}{}
For every function $f : \zo^n \to \zo$, we have $D(f) \le O(R(f)^3)$. 
\end{Nisan}

To align and compare~\hyperlink{nisan-anchor}{Nisan's Theorem} with~\Cref{fact:yao}, we restate it as follows: 

\bigskip
\hypertarget{nisan-restatement-anchor}{}
\noindent{{\bf Nisan's Theorem, restated.}}  {\it Let $R : \zo^n \times \zo^m\to \zo$ be a randomized $q$-query algorithm that computes $f : \zo^n \to \zo$ with bounded error.   There exists a deterministic $O(q^3)$-query algorithm $D : \zo^n \to \zo$ that computes $f$ exactly: $D(x) = f(x)$ for all $x\in \zo^n$.}
\bigskip

Interestingly, unlike most proofs of such relationships between query complexity measures, Nisan's proof is non-constructive.  Indeed, Nisan himself remarked: ``This result is particularly surprising as it is not achieved by simulation"~\cite[p.~329]{Nis89}. 

This non-constructive aspect of Nisan's proof was further highlighted in the work of Impagliazzo and Naor~\cite{IN88}, who sought a constructive version to derive consequences  the Turing machine model of computation.  \cite{IN88} essentially overcame this issue of non-constructivity with the added assumption that $\mathsf{P}=  \mathsf{NP}$.  In~\Cref{sec:Turing}, we discuss the implications of our constructivization of Nisan's theorem for derandomization in the Turing machine model, and compare them with the result of~\cite{IN88}.

\section{Our results: Constructive derandomization of query algorithms} 
\label{sec:our-work}

From both an algorithmic and complexity-theoretic point of view, it is natural to seek {\sl constructive} versions of~\Cref{fact:yao} and~\hyperlink{nisan-anchor}{Nisan's Theorem}: 

\begin{itemize}[leftmargin=0.5cm]
\item[$\circ$] \violet{{\sl Constructive version of~\Cref{fact:yao}:}} Given the description of a randomized $q$-query algorithm~$R$, can we {\sl efficiently construct} an deterministic $O(q/\eps)$-query algorithm $D$ that $\eps$-approximates $R$? 
\item[$\circ$] \violet{{\sl Constructive version of~\hyperlink{nisan-anchor}{Nisan's Theorem}:}}  Given the description of a randomized $q$-query algorithm that computes a function $f : \zo^n \to \zo$ with bounded error, can we {\sl efficiently construct} a deterministic $O(q^3)$-query algorithm $D$ that computes $f$ exactly?  
\end{itemize} 

 In addition to being an independently interesting aspect of derandomization to study, as alluded to in the introduction, constructivity is also the key criterion that connects derandomization in the query model of computation (a non-uniform model) to derandomization in the Turing machine model of computation (a uniform model). 

Prior work of Zimand~\cite{Zim07} and Shaltiel~\cite{Sha11} gave constructive versions of (a variant of)~\Cref{fact:yao}.  As for~\hyperlink{nisan-anchor}{Nisan's Theorem}, to our knowledge there were no known unconditional constructive versions of it; Impagliazzo and Naor~\cite{IN88} gave a constructivization under the assumption that $\mathsf{P} = \mathsf{NP}$.  We will give a detailed comparison between our results and those of~\cite{Zim07,Sha11} and~\cite{IN88} in this section.

\pparagraph{Structure of this section.}  Paralleling the structure of~\Cref{sec:background} and the two strands of our work as outlined in~\Cref{two-strands}, this section is structured as follows:

\begin{itemize}[leftmargin=0.5cm]
\item[$\circ$] In~\Cref{sec:constructive-yao} we consider general randomized query algorithms, with the goal of obtaining a constructive version of~\Cref{fact:yao}.   
\item[$\circ$] In~\Cref{sec:constructive-nisan} we consider randomized query algorithms for that compute functions $f : \zo^n \to \zo$ with bounded error, with the goal of obtaining a constructive version of~\hyperlink{nisan-anchor}{Nisan's Theorem}.   
In~\Cref{sec:Turing} we discuss the consequences of our constructivization of~\hyperlink{nisan-anchor}{Nisan's Theorem} for the Turing machine model of computation. 
\end{itemize}

In both cases, we further give {\sl instance-optimal} derandomizations: for any randomized query algorithm $R$, the deterministic query algorithm that we construct has query complexity that not only matches the bounds guaranteed by~\Cref{fact:yao} or~\hyperlink{nisan-anchor}{Nisan's Theorem}, but is in fact minimal for this specific~$R$.

\subsection{Constructive versions of~\Cref{fact:yao}}
\label{sec:constructive-yao}

Our first result is a constructive version of~\Cref{fact:yao}: 

\begin{theorem}[Constructive version of~\Cref{fact:yao}]
\label{thm:derand-alg} 
There is a deterministic algorithm $\mathcal{A}$ with the following guarantee.  Given as input a randomized $q$-query algorithm $R : \zo^n \times \zo^m \to [0,1]$ with description length $N$ and an error parameter $\eps \in (0,\frac1{2})$,  this algorithm $\mathcal{A}$ runs in 
\[ \poly(N)\cdot 2^{O(q/\eps)} \] 
 time and returns a deterministic $O(q/\eps)$-query algorithm $D : \zo^n \to [0,1]$ satisfying 
 \begin{equation}  \Ex_{\bx}\Big[\big(D(\bx) - \Ex_{\br} [R(\bx, \br)]\big)^2\Big] \le \eps. \label{eq:eps-approximating2} 
 \end{equation}
\end{theorem} 
The query complexity of $D$ matches the guarantee of~\Cref{fact:yao}, and is optimal by~\Cref{fact:lower-bound}.  The runtime of $\mathcal{A}$ is near-optimal: runtime $N + 2^{\Omega(q/\eps)}$ is necessary, since it takes time $N$ to read the description of $R$, and there are many examples of deterministic $\Theta(q/\eps)$-query algorithms $D$ that have description length $2^{\Omega(q/\eps)}$ (e.g.~the example of~\Cref{fact:lower-bound}).

As mentioned above, Zimand~\cite{Zim07} and Shaltiel~\cite{Sha11} considered the problem of constructivizing (a variant of)~\Cref{fact:yao}.  We discuss the results of~\cite{Zim07,Sha11} and compare them with~\Cref{thm:derand-alg} in~\Cref{sec:ZS}.

\subsubsection{Instance-optimal and online derandomization} 
\label{sec:extensions} 

With~\Cref{thm:derand-alg} in hand, we further consider two extensions of the basic problem of constructive derandomization: 
\begin{itemize}[leftmargin=0.5cm]
\item[$\circ$] {\sl Instance optimal derandomization:}  For any randomized $q$-query algorithm $R$, return a deterministic $q^\star_R$-query algorithm $D$, where $q^\star_R$ is minimum query complexity of any deterministic algorithm that $\eps$-approximates $R$.  By~\Cref{fact:yao} we have that $q^{\star}_R \le O(q/\eps)$, but $q^\star_R$ can in general be much smaller than $O(q/\eps)$.

Instance optimality has emerged as an influential notion in modern algorithmic research~\cite{FLN03,VV17}, as part of a broad effort to develop general frameworks for going beyond worst-case analysis~\cite{Rou19}. 
\item[$\circ$] {\sl Online derandomization:}   The algorithm of~\Cref{thm:derand-alg} constructs a deterministic query algorithm $D$ that can then be evaluated on any input $x$ of our choice.  What if we are only interested in a specific input $\underline{x}$?  Can we deterministically approximate $\E_{\br}[R(\underline{x},\br)]$, in time that is faster than constructing $D$ in its entirety and then evaluating $D$ on $\underline{x}$? 
\end{itemize}

As our algorithm for~\Cref{thm:derand-alg} does not seem to be amendable to either of the above extensions, we develop new techniques and fundamentally different algorithms to achieve them.  These techniques turn out to be of interest and utility beyond the specific applications above: for our instance-optimal derandomization algorithm, we develop a general framework that we will later on also use to derive an instance-optimal constructivization of~\hyperlink{nisan-anchor}{Nisan's Theorem}.   For  our online derandomization algorithm, we generalize the powerful O'Donnell--Saks--Schramm--Servedio inequality~\cite{OSSS05} from deterministic to randomized query algorithms. 

\pparagraph{An instance-optimal algorithm.}  We begin by describing our instance-optimal algorithm.

\begin{notation}[$q^\star_R$]
\label{notation:minimal-general} 
Let $R : \zo^n \times \zo^m \to [0,1]$ be a randomized query algorithm.  We write $q^{\star}_R$ to denote the \emph{minimum query complexity of any deterministic algorithm that $\eps$-approximates $R$}:  
\[ q^\star_R \coloneqq \{\,q' \colon \text{there is a $q'$-query $\DDT$ $D$ that $\eps$-approximates $R$\,}\}. \] 
\end{notation}

\begin{theorem}[Instance-optimal derandomization] 
\label{thm:instance-opt} 
There is a deterministic algorithm $\mathcal{A}_{\mathrm{InstanceOpt}}$ with the following guarantee.  Given as input a randomized $q$-query algorithm $R : \zo^n \times \zo^m \to [0,1]$ with description length $N$ and an error parameter $\eps \in (0,\frac1{2})$, this algorithm $\mathcal{A}_{\mathrm{InstanceOpt}}$ runs in 
\[ \poly(N)\cdot n^{O(q^{\star}_R)} \] 
 time and returns a deterministic $q^{\star}_R$-query algorithm $D : \zo^n \to [0,1]$ satisfying
  \[  \Ex_{\bx}\Big[\big(D(\bx) - \Ex_{\br} [R(\bx, \br)]\big)^2\Big] \le \eps. \] 
 \end{theorem} 


As alluded to above, we derive~\Cref{thm:instance-opt} as a corollary of a general framework that we develop for achieving instance-optimality in the derandomization of query algorithms with respect to a broad class of error metrics: 

\begin{theorem}[General framework for instance-optimal  derandomization; informal version]
\label{thm:framework} 
Let $\mathcal{E} : \{ \RDT\text{s}\} \times \{ \DDT\text{s}\} \to [0,1]$ be a ``\,$t$-efficient" error metric for measuring the distance between $\RDT$s and $\DDT$s.  There is a deterministic algorithm, $\mathcal{A}_{\mathrm{InstanceOpt},\mathcal{E}}$ with the following guarantee: Given as input a $q$-query $\RDT$ $R$ with description length $N$ and an error parameter $\eps \in (0,1)$, for
    \[ q^\star_{R,\mathcal{E}} \coloneqq \min \{\,q' \colon \text{there is a $q'$-query $\DDT$ $D$ such that $
    \mathcal{E}(R,D) \leq \epsilon$\,}\}. \] 
    $\mathcal{A}_{\mathrm{InstanceOpt},\mathcal{E}}$ runs in
    \[ \poly(N,t, n^{q^{\star}_{R,\mathcal{E}}})  \] 
    time and returns a $q^{\star}_{R,\mathcal{E}}$-query $\DDT$ $D$ satisfying $\mathcal{E}(R,D) \leq \epsilon$.   
\end{theorem}

\Cref{thm:instance-opt} follows as an immediate corollary of~\Cref{thm:framework} by instantiating it with the error metric $\mathcal{E}$ being $L_2$ error.  The framework of~\Cref{thm:framework} is fairly versatile:  in~\Cref{sec:constructive-nisan} we will see that it also yields an instance-optimal constructivization of~\hyperlink{nisan-anchor}{Nisan's Theorem} (though this application will require choosing the error metric $\mathcal{E}$ carefully and involve more technical work).

\pparagraph{An online algorithm.} Our online algorithm as follows: 

\begin{theorem}[Online derandomization] 
\label{thm:online} 
There is a deterministic algorithm $\mathcal{A}_{\mathrm{Online}}$ with the following guarantee.  Given as input a randomized $q$-query algorithm $R : \zo^n \times \zo^m \to [0,1]$ with description length $N$, an error parameter $\eps \in (0,\frac1{2})$, and an input $\underline{x} \in \zo^n$, this algorithm $\mathcal{A}_{\mathrm{Online}}$ runs in 
\[ \poly(N, q,1/\eps) \] 
 time, makes $O(q^2/\eps^3)$ queries to $\underline{x}$,   and returns a value $\mathcal{A}_{\mathrm{Online}}(\underline{x}) \in [0,1]$.  The output values of $\mathcal{A}_{\mathrm{Online}}$ satisfy:  
 \[ \Ex_{\bx}\Big[ \big(\mathcal{A}_{\mathrm{Online}}(\bx) - \Ex_{\br} [R(\bx, \br)]\big)^2\Big] \le \eps. \] 
\end{theorem}

The key qualitative advantage of~\Cref{thm:online} is that $\mathcal{A}_{\mathrm{Online}}$'s runtime is polynomial in all the relevant parameters.  Such a runtime is achievable because we are considering the online version of the problem, where the derandomization algorithm is not expected to return the entire description of the deterministic query algorithm $D$.  We can think of $\mathcal{A}_{\mathrm{Online}}$ as constructing just one branch of $D$: the branch that $\underline{x}$ is consistent with. 

Our algorithm $\mathcal{A}_{\mathrm{Online}}$ and its analysis build on the work of Aaronson and Ambainis~\cite{AA14}, who were interested in {\sl quantum} query algorithms. Recalling~\Cref{remark:AA}, the work of~\cite{AA14} was motivated by the possibility of a quantum analogue of~\Cref{fact:yao}: showing---even just non-constructively---that the acceptance probabilities of a quantum query algorithm $Q$ can be approximated on most inputs by a deterministic query algorithm $D$ (whose query complexity is polynomially related to that of $Q$'s).  \violet{In~\cite{AA14}, the authors posed a Fourier-analytic conjecture about the influence of variables in bounded low-degree polynomials $p : \zo^n \to [0,1]$, and showed that this conjecture would yield a quantum analogue of~\Cref{fact:yao}.  In fact, assuming this Fourier-analytic conjecture, their proof of the quantum analogue of~\Cref{fact:yao} is  even constructive, where the meta-algorithm that constructs $D$ is efficient if $\mathsf{P} = \mathsf{P}^{\#\mathsf{P}}$}.
 This conjecture is now known as the {\sl Aaronson--Ambainis conjecture}, and it remains a major open problem in the analysis of boolean functions~\cite{Sim14}.   

\violet{The first} ingredient in our proof of~\Cref{thm:online} is a lemma showing that the Aaronson--Ambainis conjecture holds for randomized query algorithms: 

\begin{lemma}[Every randomized query algorithm has an influential variable] 
\label{lem:OSSS-RDT}
Let $R : \zo^n \times\zo^m \to [0,1]$ be a randomized $q$-query algorithm and consider its mean function $\mu_R(x) \coloneqq \Ex_{\br}[R(x,\br)].$\footnote{To see the connection to the Aaronson--Ambainis conjecture, note that $\mu_R : \zo^n\to [0,1]$ is a polynomial of degree at most $q$.}
There is a variable $i\in [n]$ such that 
\[ \Inf_i(\mu) \coloneqq \Prx_{\bx}\big[|\mu_R(\bx)-\mu_R(\bx^{\oplus i})|\big] \ge \frac{\Var(\mu_R)}{q},\]
where $\bx^{\oplus i}$ denotes $\bx$ with its $i$-th coordinate flipped.  
\end{lemma} 

\Cref{lem:OSSS-RDT} is in turn a generalization of the analogous inequality for {\sl deterministic} query algorithms, a powerful result due to O'Donnell, Saks, Schramm, and Servedio~\cite{OSSS05}.  We show that~\Cref{lem:OSSS-RDT} is a straightforward consequence of a ``two-function version" of the~\cite{OSSS05} inequality; this two-function version is also due to~\cite{OSSS05}.

\violet{The second ingredient in our proof is a modification of~\cite{AA14}'s algorithm and analysis to remove their assumption of $\mathsf{P} = \mathsf{P}^{\#\mathsf{P}}$ in the case of randomized query algorithms.  In~\cite{AA14}'s analysis, this assumption underlies their design of an efficient deterministic algorithm for computing the influence of variables within quantum query algorithms.  We give an unconditional, efficient algorithm in the case of randomized query algorithms.}

\subsubsection{Comparison of~\Cref{thm:derand-alg,thm:instance-opt,thm:online}}  While both~\Cref{thm:instance-opt,thm:online} improve upon~\Cref{thm:derand-alg} in qualitative ways, neither strictly improves upon~\Cref{thm:derand-alg}.  The runtime of $\mathcal{A}_{\mathrm{InstanceOpt}}$ from~\Cref{thm:instance-opt} is $\poly(N)\cdot n^{O(q^\star_R)}$, which is incomparable to the runtime of  $\mathcal{A}$ from~\Cref{thm:derand-alg} ($\poly(N)\cdot 2^{O(q/\eps)}$).  The algorithm $\mathcal{A}_{\mathrm{Online}}$ of~\Cref{thm:online} has query complexity $O(q^2/\eps^3)$, whereas the algorithm of~\Cref{thm:derand-alg} returns $D$ with query complexity $O(q/\eps)$.  \violet{The possibility of designing a unified algorithm that achieves the ``best of all worlds" is an interesting avenue for future work.}

\subsection{Constructive version of Nisan's theorem} 
\label{sec:constructive-nisan} 

We now turn to the second strand of our work (as described on~\cpageref{two-strands}): we consider the special case of randomized query algorithms for {\sl total decision problems} and the problem of constructivizing~\hyperlink{nisan-anchor}{Nisan's Theorem}.  Recall that~\hyperlink{nisan-anchor}{Nisan's Theorem} establishes the existence of a {\sl zero-error} derandomization of randomized $q$-query algorithms that solve total decision problems $f : \zo^n \to \zo$ with bounded error: it establishes the existence of a deterministic $O(q^3)$-query algorithm that computes $f$ exactly.

Using the general framework we developed for proving~\Cref{thm:instance-opt} (\Cref{thm:framework}), we obtain the following instance-optimal constructivization of~\hyperlink{nisan-anchor}{Nisan's Theorem}.    In this context, the corresponding notion of minimal deterministic query complexity is the following: 

\begin{notation}[$q^\star_R$]
\label{notation:minimal-2}
Let $R : \zo^n \times \zo^m \to \zo$ be a randomized query algorithm that computes $f : \zo^n \to \zo$ with bounded error.  We write $q^{\star}_R$ to denote the \emph{minimum query complexity of any deterministic algorithm that computes $f$ exactly}:  
\[ q^\star_R \coloneqq \{\,q' \colon \text{there is a $q'$-query $\DDT$ $D$ such that $D(x) = f(x)$ for all $x$\,}\}. \] 
\end{notation}

\begin{restatable}[Instance-optimal constructivization of~\hyperlink{nisan-anchor}{Nisan's Theorem}]{theorem}{nisan}
\label{thm:constructive-Nisan} 
There is a deterministic algorithm $\mathcal{A}_{\mathrm{Nisan}}$ with the following guarantee.  Given as input a randomized $q$-query algorithm $R$ with description length $N$ that computes function $f : \zo^n \to \zo$ with bounded error, this algorithm $\mathcal{A}_{\mathrm{Nisan}}$ runs in 
\[ \poly(N)\cdot n^{O(q^3)} \] 
 time and returns a $q^{\star}_R$-query deterministic decision tree $D : \zo^n \to \zo$ that computes $f$ exactly: $D(x) = f(x)$ for all $x \in \zo^n$.   
\end{restatable} 


To our knowledge, prior to our work there were no known constructivizations of~\hyperlink{nisan-anchor}{Nisan's Theorem}, even one with just a worst-case bound on the query complexity of $D$ rather than an instance-optimal one (i.e.~a bound of $O(q^3)$ as guaranteed by~\hyperlink{nisan-anchor}{Nisan's Theorem}, rather than $q^\star_R$).  Indeed, Nisan himself remarked ``This result is particularly surprising as it is not achieved by simulation"~\cite[p.~329]{Nis89}.  This non-constructive aspect of Nisan's proof was further highlighted in the work of Impagliazzo and Naor~\cite{IN88}, who sought a constructive version to derive consequences  the Turing machine model of computation; we discuss the work of~\cite{IN88} in the next subsection.

\subsubsection{Consequences for derandomizing Turing machine computation} 
\label{sec:Turing} 

Constructivity is the key criterion that connects derandomization in the query model of computation, a non-uniform model, to derandomization in the Turing machine model of computation, a uniform model.  The following is a straightforward corollary of our constructivization (\Cref{thm:constructive-Nisan}) of~\hyperlink{nisan-anchor}{Nisan's Theorem}.  
  (As is standard when reasoning about sublinear-time computation, we consider {\sl random access} Turing machines.)

\begin{restatable}[Uniform derandomization with preprocessing]{corollary}{derandomPreprocess}
\label{cor:derandomization-with-preprocessing}  If $L \sse \zo^*$ is a language decided by a $\polylog(n)$-time randomized Turing machine (allowing for two-sided error), then $L$ is also decided by a $\polylog(n)$-time deterministic Turing machine with a $\quasipoly(n)$-time preprocessing step, a one-time cost for all inputs of length $n$. 
\end{restatable}

\Cref{cor:derandomization-with-preprocessing} can be expressed succinctly as: 
\[ \mathsf{BPTIME}(\polylog(n)) \sse ``\mathsf{Preprocess}(\quasipoly(n)) + \mathsf{TIME}(\polylog(n))". \] 
The same proof ``scales up" to give, say, $\mathsf{BPTIME}(n^{o(1)}) \sse \mathsf{Preprocess}(2^{n^{o(1)}}) + \mathsf{TIME}(n^{o(1)})$.

Even the following weaker version of~\Cref{cor:derandomization-with-preprocessing}, where one does not ``factor out" the preprocessing step, does not appear to have been known prior to our work.  Let $\mathsf{TIME}(t,q)$ denote the class of languages decided by a time-$t$ deterministic Turing machine that makes $q$-queries to the input.  Then 
\begin{equation} \mathsf{BPTIME}(\polylog(n)) \sse \mathsf{TIME}(\quasipoly(n),\polylog(n)). \label{eq:us}
\end{equation} 

\pparagraph{Comparision with naive constructivizations.}  There are two easy ways to constructively derandomize $\mathsf{BPTIME}(\polylog(n))$. One is to try all possible $\polylog(n)$-query deterministic algorithms, of which there are $n^{\quasipoly(n)}$ many. This implies that:
\begin{equation}
     \mathsf{BPTIME}(\polylog(n)) \sse \mathsf{TIME}(n^{\quasipoly(n)},\polylog(n)). \label{eq:them1} 
\end{equation}
A second naive algorithm would be, on an input $x$, to try all possible $2^{\polylog(n)}$ random strings and return the majority output. These different choices of the random string might result in queries to different coordinates of the input, meaning that up to $n$ coordinates can be queried, the trivial number.  Hence: 
\begin{equation}
    \mathsf{BPTIME}(\polylog(n)) \sse \mathsf{TIME}(\quasipoly(n), n). \label{eq:them2} 
\end{equation}
Our result (\ref{eq:us}) can therefore be viewed as achieving the best of both worlds (\ref{eq:them1}) and (\ref{eq:them2}).

\pparagraph{Comparison with Impagliazzo--Naor~\cite{IN88}.}  The connection between~\hyperlink{nisan-anchor}{Nisan's Theorem} and the derandomization of sublinear-time Turing machine computation, and the challenges posed by the non-constructive nature of Nisan's proof, were highlighted in the work of Impagliazzo and Naor~\cite{IN88}.  This work essentially overcame the issue of non-constructivity with the added assumption that $\mathsf{P} = \mathsf{NP}$:

\begin{theorem}[\cite{IN88}]
\label{thm:IN88}
If $\mathsf{P} = \mathsf{NP}$ then $\mathsf{BPTIME}(\polylog(n)) = \mathsf{TIME}(\polylog(n))$. 
\end{theorem}

(\Cref{thm:IN88} can be viewed as a strengthening of a basic and classical result of structural complexity theory:  if $\mathsf{P} = \mathsf{NP}$ then $\mathsf{BPP} = \mathsf{P}$.) While the conclusion of~\Cref{thm:IN88} is stronger than our~\Cref{cor:derandomization-with-preprocessing}, it only holds under the assumption that $\mathsf{P} = \mathsf{NP}$, whereas~\Cref{cor:derandomization-with-preprocessing} is  unconditional.

\subsubsection{Consequences for dequantizing Turing machine computation} 

As a further application of our framework (\Cref{thm:framework}), we show that it can be used to constructivize yet another a classic result in query complexity, this time relating the deterministic query complexity of a total boolean function $f : \zo^n \to \zo$ to its (bounded-error) {\sl quantum query complexity}.  The following theorem is due to Beals, Burhman, Cleve, Mosca, de Wolf~\cite{BBCMdW01}:

\hypertarget{beals-anchor}{}
\begin{theorem}[Quantum versus deterministic query complexity]
\label{thm:beals}
For every $f : \zo^n \to \zo$, we have that $D(f) \le O(Q(f)^6)$.
\end{theorem}

Given the description of a quantum query algorithm for a function $f$,~\Cref{thm:framework} can be used to find a deterministic algorithm with minimal query complexity computing $f$ exactly (and by~\Cref{thm:beals}, we are guaranteed that this query complexity is at most $O(Q(f)^6)$).  Like our constructivization of~\hyperlink{nisan-anchor}{Nisan's Theorem}, this has immediate implications for computation in the Turing machine model;  the following is a quantum analogue of~\Cref{cor:derandomization-with-preprocessing}: 

\begin{restatable}[Uniform dequantization with preprocessing]{corollary}{dequantizePreprocess}
\label{cor:dequantization-with-preprocessing}  If $L \sse \zo^*$ is a language decided by a $\polylog(n)$-time $m$-qubit quantum Turing machine (allowing for two-sided error), then $L$ is also decided by a $\polylog(n)$-time deterministic Turing machine with a $\poly(2^m)\cdot \quasipoly(n)$-time preprocessing step, a one-time cost for all inputs of length $n$.   
\end{restatable}

\subsection{Recap and summary of our techniques} 
\label{sec:techniques}

Recapping and summarizing the discussion in our introduction, in this work we draw on a range of techniques to prove our results: 

\begin{table}[H]
\begin{adjustwidth}{-2em}{}
\renewcommand{\arraystretch}{1.7}
\centering
\begin{tabular}{|c|>{\centering\arraybackslash}p{8cm} |}
\hline
  Result & Techniques \\ \hline \hline 
  \Cref{thm:derand-alg} &~~~ PRGs and randomness samplers ~~~\\ \hline 
~~\Cref{thm:instance-opt,thm:constructive-Nisan}~~ & Instance-optimal framework  (\Cref{thm:framework}) \\ \hline
\Cref{thm:online} & Greedy top-down algorithm + \Cref{lem:OSSS-RDT} \\ \hline 
\end{tabular}
\label{table:techniques}
\end{adjustwidth}
\end{table}
\vspace{-7pt}

\begin{itemize}[leftmargin=0.5cm]
\item[$\circ$] Our algorithm for~\Cref{thm:derand-alg} and its analysis are both quite simple.  We first use two basic pseudorandomness constructs---pseudorandom generators and randomness samplers---to deterministically construct a small list of candidate $\eps$-approximating deterministic query algorithms.  We are then faced with the question: given a randomized query algorithm $R$ and a deterministic query algorithm $D$, can one efficiently and {\sl deterministically} compute their distance $\Ex_{\bx}[(D(\bx)-\Ex_{\br}[R(\bx,\br)])^2]$?  We solve this problem using elementary Fourier analysis of boolean functions. 
\item[$\circ$]  As described in the introduction, to prove~\Cref{thm:instance-opt,thm:constructive-Nisan} we develop a general framework,~\Cref{thm:framework}, for achieving instance-optimal derandomization of randomized query algorithms with respect to a broad class of error metrics. \Cref{thm:instance-opt} follows as an immediate corollary of this framework by taking the error metric to be $L_2$ distance.  For our constructivization of~\hyperlink{nisan-anchor}{Nisan's Theorem} and~\hyperlink{beals-anchor}{Beals et al.'s Theorem}, we invoke this framework with other carefully chosen error metrics.
\item[$\circ$] Our proof of \Cref{thm:online} draws on a powerful result from concrete complexity: every small-depth deterministic decision tree has an ``influential" variable~\cite{OSSS05}.  Our key lemma here shows that the~\cite{OSSS05} inequality also holds for {\sl randomized} decision trees. With this generalization in hand, we then analyze the following natural online algorithm:  on input $\underline{x}$, query $\underline{x}_i$ where $i$ is the most influential variable of $R$; restrict $R$ accordingly, and recurse.  \violet{While~\cite{AA14} had shown that the influence of variables within quantum query algorithms can be deterministically and efficiently computed under the assumption that $\mathsf{P} = \mathsf{P}^{\#\mathsf{P}}$, we give an unconditional, efficient algorithm in the case of randomized query algorithms.} 
\end{itemize}

\subsection{The works of Zimand and Shaltiel} 
\label{sec:ZS} 
In this section we compare~\Cref{thm:derand-alg} to prior work of Zimand~\cite{Zim07} and Shaltiel~\cite{Sha11}.  The following is a variant of~\Cref{fact:yao}: 

\begin{fact}
\label{fact:yao2} 
Let $R : \zo^n \times \zo^m \to \zo$ be a randomized $q$-query algorithm satisfying 
\begin{equation} \Prx_{\bx,\br}[R(\bx,\br) \ne f(\bx)] \le \delta \qquad \text{for some $f : \zo^n \to \zo$.}\label{eq:ZS-assumption}
\end{equation} 
There exists a deterministic $q$-query algorithm $D$ satisfying $\Prx_{\bx}[D(\bx)\ne f(\bx)] \le \delta$. 
\end{fact}

Like~\Cref{fact:yao}, the proof of~\Cref{fact:yao2} is a straightforward application of the easy direction of Yao's lemma, and is therefore also non-constructive.  Zimand~\cite{Zim07} and Shaltiel~\cite{Sha11} considered the problem of constructivizing~\Cref{fact:yao2}.   Zimand proves the following: 
\begin{theorem}[\cite{Zim07}]
\label{thm:zimand}
There is an absolute constant $\alpha < 1$ such that the following holds.  Let $R : \zo^n \times \zo^m \to \zo$ be an explicitly constructible\footnote{A $q$-query algorithm is {\sl explicitly constructible} if there is a polynomial-time Turing machine which, when given the answers to the queries made so far, computes the next query in time $\poly(q, \log n)$.  For randomized query algorithms, the machine also receives a string $r \in \zo^m$ where $m$ is the randomness complexity of the algorithm.}  randomized query algorithm for $f :\zo^n \to \zo$ satisfying (\ref{eq:ZS-assumption}) with $\delta \le \frac1{3}$.   Suppose that the randomness complexity of $R$ is $q \le n^{\alpha}$ and its randomness complexity is $m \le q$.  Then  there is an explicitly constructible deterministic $O(q^{24})$-query algorithm $D$ such that  $\Prx_{\bx}[D(\bx) \ne f(\bx)]\le O(\delta)$.
\end{theorem} 

Shaltiel gives the following improvement of Zimand's result: 

\begin{theorem}[\cite{Sha11}]
\label{thm:shaltiel}
There is an absolute constant $\beta < 1$ such that the following holds.  Let $R : \zo^n \times \zo^m \to \zo$ be an explicitly constructible randomized query algorithm for $f :\zo^n \to \zo$ satisfying (\ref{eq:ZS-assumption}) with $\delta \le \frac1{3}$.   Suppose the query and randomness complexities of $R$ satisfy $q + m \le \beta n$.  Then  there is an explicitly constructible deterministic $O(q+m)$-query algorithm $D$ such that  $\Prx_{\bx}[D(\bx) \ne f(\bx)]\le O(\delta)$.
\end{theorem} 

We remark that~\Cref{thm:shaltiel} is just one of many results in~\cite{Sha11}, which considers the problem of constructive derandomization in a number of other computational models (communication complexity, streaming, constant-depth circuits, etc.) in addition to the query model.

\paragraph{Comparing our result (\Cref{thm:derand-alg}) to Zimand's and Shaltiel's (\Cref{thm:zimand,thm:shaltiel}).}

\begin{itemize}[leftmargin=0.5cm]
\item[$\circ$]  First, there is a high-level difference in terms of the overall setup: we assume that the derandomizing algorithm is given $R$ as input, and it is then expected to output the description of $D$; in~\cite{Zim07,Sha11}, $R$ is assumed to be explicitly constructible, and these works show that $D$ is also explicitly constructible.  Note that if a query algorithm $D$ is explicitly constructible, then its description can be printed in time $|D|\cdot \poly(q,\log n)$, where $|D|$ denotes the description length of $D$.  In this regard the results of~\cite{Zim07,Sha11} are stronger than ours.

\item[$\circ$] In the results of~\cite{Zim07,Sha11}, the query complexity of the resulting deterministic algorithm $D$ depends on the randomness complexity `$m$' of $R$, whereas~\Cref{thm:derand-alg} does not.  In~\Cref{thm:shaltiel} (\cite{Sha11}'s result) the query complexity of $D$ is $O(q+m)$, and in~\Cref{thm:zimand} (\cite{Zim07}'s result), $m$ is restricted to be at most $q$ to begin with.   In contrast, the query complexity of $D$ in~\Cref{thm:derand-alg} is $O(q/\eps)$ regardless of the value of~$m$.  We note that there are simple examples of randomized query algorithms for which $m \gg q$ (e.g.~the example in~\Cref{fact:pointwise-impossible} where $q=1$ and $m = \log n$).

\item[$\circ$] \Cref{thm:derand-alg} applies to general randomized query algorithms $R$ (with no assumptions about the distribution of $R(x,\br)$), and returns a deterministic $D$ that approximates $R$'s acceptance probabilities.  The results of~\cite{Zim07,Sha11} focus on $R$'s that satisfy (\ref{eq:ZS-assumption}), and return a $D$ that achieving a similar guarantee.

\item[$\circ$] The proofs of~\cite{Zim07,Sha11} are based on a general framework, due to Goldreich and Wigderson~\cite{GW02}, of ``derandomization by extracting randomness from the input".  (See~\cite{Sha10} for an excellent survey of this framework.) Both works use extractors within this framework to tame the correlations between the uniform random input ($\bx \sim \zo^n$) and the randomness employed by the query algorithm ($\br \sim \zo^m$): Zimand uses exposure resilient extractors, and Shaltiel uses extractors for bit-fixing sources.

As outlined in~\Cref{sec:techniques}, our approach to proving~\Cref{thm:derand-alg} is quite different from that of~\cite{Zim07,Sha11}: it is not based on the framework of~\cite{GW02} and does not involve extractors (though it does rely on other basic pseudorandomness constructs such as PRGs and randomness samplers).  
\end{itemize}

\section{Preliminaries} 
\label{sec:prelim} 

All probabilities and expectations are with respect to the uniform distribution; we use {\bf boldface} to denote random variables. 
Throughout this paper, we consider the most natural representation of query algorithms, as a {\sl binary decision tree}:
\begin{definition}[Randomized and deterministic decision trees] 
\label{def:DT} 
An $n$-variable \emph{randomized decision tree} $(\RDT)$ is a binary tree $R$ with two types of internal nodes: 
\begin{itemize} 
\item[$\circ$] \emph{Decision nodes} that branch on the outcome of boolean variables $x_1,\ldots,x_n$, 
\item[$\circ$] \emph{Stochastic nodes} that branch on the outcome of a $\mathrm{Bernoulli}(\frac1{2})$ random variable. 
\end{itemize} 
The leaves of $R$ are labelled by values in $[0,1]$.   The \emph{query complexity} of $R$ is the maximum number of decision nodes in any root-to-leaf path, and the \emph{randomness complexity of} of $R$ is the maximum number of stochastic nodes in any root-to-leaf path.   Please see~\Cref{fig:RDT-figure}.

A \emph{deterministic} decision tree $(\DDT)$ is a randomized decision tree with no stochastic nodes.  
\end{definition} 

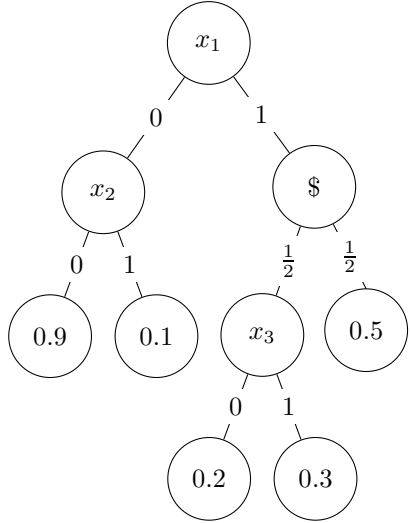
\begin{figure}[h]
\forestset{
  triangle/.style={
    node format={
      \noexpand\node [
      draw,
      shape=regular polygon,
      regular polygon sides=3,
      inner sep=0pt,
      outer sep=0pt,
      \forestoption{node options},
      anchor=\forestoption{anchor}
      ]
      (\forestoption{name}) {\foresteoption{content format}};
    },
    child anchor=parent,
  }
 }
\begin{align*}
    \hspace{3mm}{\small
    \!\begin{gathered}
    \begin{forest}
    for tree={
        grow=south,
        circle, draw, minimum size=11mm, l sep = 8mm, s sep = 3mm, inner sep = 1mm
    }
    [$x_1$
        [$x_2$, edge label = {node [midway, fill=white] {$0$} }
            [$0.9$, edge label = {node [midway, fill=white] {$0$} }]
            [$0.1$, , edge label = {node [midway, fill=white] {$1$} }]
        ]
        [$\$$, edge label = {node [midway, fill=white] {$1$} }
            [$x_3$, edge label = {node [midway, fill=white] {$\frac1{2}$} }
                [$0.2$, edge label = {node [midway, fill=white] {$0$} }]
                [$0.3$, edge label = {node [midway, fill=white] {$1$} }]
            ]
            [$0.5$, triangle, edge label = {node [midway, fill=white] {$\frac1{2}$} }]
        ]
    ]
    \end{forest}
    \end{gathered}}
\end{align*}
\caption{A randomized decision tree ($\RDT$)}
\label{fig:RDT-figure} 
\end{figure}

\begin{notation*} 
Let $R : \zo^n \times \zo^m \to [0,1]$ be a $q$-query $\RDT$.  For each $r \in \zo^m$, we define the function $R_r(x) \coloneqq R(x,r)$, and note that $R_r$ is a $q$-query $\DDT$. 
\end{notation*}

\vspace{-10pt} 

\pparagraph{Decision trees and the functions they compute.}  Every randomized decision tree can be associated with a randomized function that it computes, which we will express as $R : \zo^n \times \zo^m \to [0,1]$, where $m$ is its randomness complexity: on input $x \in \zo^n$, the output of $R$ is the random variable 
\[ R(x,\br)\quad \text{where $\br \sim \zo^m$ is uniform random.} \]    We also associate with $R$ its {\sl mean function} $\mu_R : \zo^n \to [0,1]$, where 
\[ \mu_R(x) = \Ex_{\br}[R(x,\br)].  \]  
Given two $\RDT$s $R_1$ and $R_2$, we say that $R_2$ $\eps$-approximates $R_1$ if $\| \mu_{R_1}-\mu_{R_2} \|_2^2 \le \eps$.  We will most often (though not always) use this terminology with $R_2$ being a $\DDT$.

\pparagraph{Decision trees and total decision problems.} We will also be interested in the special case of randomized decision trees that solve {\sl total decision problems} with bounded error:   

\begin{definition}[Bounded-error $\RDT$s for total decision problems] 
\label{def:decision-RDTs}
Let $f : \zo^n \to \zo$ be a boolean function and $R : \zo^n \times \zo^m \to \zo$ be an $\RDT$.  We say that $R$ is an \emph{$\RDT$ that computes $f$ with bounded error} if 
\[ \text{For all $x\in \zo^n$}, \quad \mu_R(x) = 
\begin{cases} 
[\frac{2}{3},1] & \mathrm{if }\ f(x) = 1 \\
[0,\frac{1}{3}] & \mathrm{if }\ f(x) = 0.
\end{cases} 
\] 
We write $R(f)$ to denote the \emph{randomized decision tree complexity} of $f$, 
\[ R(f) \coloneqq \min\{\,q \colon \text{there is a $q$-query $\RDT$ that computes $f$ with bounded error}\,\}, \]
and likewise $D(f)$ to denote its \emph{deterministic decision tree complexity}.
\end{definition}

\section{Proof of \Cref{thm:derand-alg}}

Our algorithm will have two conceptual steps:
\begin{enumerate}
    \item We first deterministically generate a list of not-too-many candidate $O(q/\eps)$-query $\DDT$s, with the guarantee that at least one of which must be a $\eps$-approximation of the $\RDT$ $R$.
    \item We show how to deterministically and efficiently compute the $L_2$ error $\| D-\mu_R \|_2^2$ between a $\DDT$ $D$ and $\RDT$ $R$,  allowing us to identify a candidate that is a $\varepsilon$-approximation the $\RDT$.
\end{enumerate}
\subsection{Step 1: Deterministically generating a list of candidates} 

If we do not care about the number of candidates returned, the first step is easily accomplished by applying the the algorithm implicitly defined by the proof of \Cref{fact:yao}. In that proof, we guarantee there is at least one outcome $(r_1,\ldots,r_{c})$ of $c = 1/\eps$ random strings $\br_1,\ldots,\br_c\sim \zo^m$ that can be used to construct a $O(q/\eps)$-query $\DDT$ that is an $\varepsilon$-approximation of the $\RDT$ $R$. Unfortunately, there are $2^{O(m/\eps)}$ possible candidates, and going through all of them---even assuming we can accomplish Step 2 of identifying a good candidate---would be much too slow.

In order to make this more efficient, we make the following two optimizations.
\begin{enumerate}
    \item[a.] We first use a pseudorandom generator to deterministically convert $R$ into another $\RDT$ $\tilde{R}$ that is an $\eps$-approximating of $R$ and has randomness complexity $\tilde{m} = O(\log(N/\eps))$.
    \item[b.] Rather than choosing $1/\eps$ many random strings independently and uniformly at random, we sample them only with pairwise independence. This is sufficient for our purposes and reduces the list of candidates from $N^{\Omega(1/\eps)}$ to $\poly(N,1/\eps)$. 
\end{enumerate}

We now formalize the above. First, we use a standard pseudorandom generator to reduce the randomness complexity of $R$:
\begin{lemma}[Randomness complexity reduction via PRGs]
\label{lemma: PRG}
    There is a deterministic algorithm that takes as input a $q$-query $\RDT$ $R: \zo^n \times \zo^m \to [0,1]$ with description length $N$ and error parameter $\varepsilon \in (0, \frac{1}{2})$, runs in $\poly(N,1/\eps)$ time and returns a $q$-query $\RDT$ $\tilde{R}: \zo^n \times \zo^{\tilde{m}} \to [0,1]$ with description length $O(N)$ and randomness complexity $\tilde{m} = O(\log(N/\eps))$ satisfying $\| \mu_R - \mu_{\tilde{R}} \|_2^2 \le \eps$. 
\end{lemma}
\begin{proof}
For any fixed $x\in \zo^n$, the function $r\mapsto R(x,r)$ size-$N$ $\RDT$ comprising only of stochastic nodes.  It is a basic fact from derandomization theory that size-$N$ decision trees can be ``$\eps$-fooled with seed length $\tilde{m} = O(\log(N/\eps))$", meaning that there is an explicit and efficiently computable function $\mathcal{G} : \zo^{\tilde{m}} \to \zo^m$ such that 
    \begin{equation} \Big| \Ex_{\br\sim \zo^m}[R(x,\br)] - \Ex_{\bs\sim \zo^{\tilde{m}}}[R(x,\mathcal{G}(\bs))]\Big| \le \eps. \label{eq:fool}
    \end{equation} 
    This follows from the fact that size-$N$ decision trees (with output values in $[0,1]$) have Fourier $L_1$ norm at most $N$~\cite{KM93}, along with standard constructions of small-bias probability spaces~\cite{NN93,AGHP92}. 
    
    We define the function $\tilde{R} : \zo^n \times \zo^{\tilde{m}}\to [0,1]$, 
    \[ \tilde{R}(x,s) \coloneqq R(x,\mathcal{G}(s)), \]
    and note that $\tilde{R}$ is a $q$-query $\RDT$ with description length $O(N)$.  Note also that the bound (\ref{eq:fool}) can be expressed as $|\mu_R(x)- \mu_{\tilde{R}}(x)| \le \eps$.  Since this holds for all $x \in \zo^n$, the lemma follows. 
    \end{proof}

Next, we show how to use samplers to efficiently generate candidates.
\begin{lemma}[A short list of candidates via pairwise independent samplers]
\label{lemma: pairwise}
    There is a deterministic algorithm that takes as input a $q$-query $\RDT$ $\tilde{R}: \zo^n \times \zo^{\tilde{m}}  \to [0,1]$ with description length $N$, runs in
    \begin{align*}
        \poly(N, 2^{q/\eps}, 2^{\tilde{m}})
    \end{align*}
    time, and returns a list of $L = 2^{O(\tilde{m})}$ many $O(q/\eps)$-query $\DDT$s $\{ D_1,\ldots, D_L \}$ such that $\| D_{i^\star} - \mu_R \|_2^2 \le \eps$ for at least one $i^{\star} \in [L]$. 
 \end{lemma}
\begin{proof}
    We use pairwise independent samplers~\cite{CG89}: this is an efficiently computable deterministic function that maps a seed of $O(\tilde{m})$ random bits into $\br_1,\ldots, \br_c \sim \zo^{\tilde{m}}$ that are pairwise independent. It is easily verified that the proof of \Cref{fact:yao} only requires $\br_1,\ldots, \br_c$ to be picked with pairwise independence (since it is based only on first and second moment calculations). Hence, we can just try all possible choices for the seed, of which there are $2^{O(\tilde{m})}$, and for each include the resulting stacked tree as a candidate.
\end{proof}

Combining the above two lemmas with triangle inequality yields the following: 
\begin{corollary}[Combining~\Cref{lemma: PRG,lemma: pairwise}]
    \label{cor:candidates}
    There is a deterministic algorithm that takes as input a $q$-query $\RDT$ $R: \zo^n \times \zo^m \to [0,1]$ with description length $N$, runs in
    \begin{align*}
        \poly(N, 2^{q/\eps})
    \end{align*}
   time and returns a list of $L = \poly(N, 1/\eps)$ many $O(q/\eps)$-query $\DDT$s $\{ D_1,\ldots, D_L\}$ such that $\| D_{i^\star}-\mu_R\|_2^2 \le 4\eps$ for at least one $i^\star\in [L]$. 
    \end{corollary}
\begin{proof}
    Using \Cref{lemma: PRG}, we first deterministically convert $R$ into  $\tilde{R}$, a $q$-query $\RDT$ that $\eps$-approximates $R$ and has randomness complexity $\tilde{m} = O(\log(N/\eps))$.   Then, we use \Cref{lemma: pairwise} to generate $L= 2^{O(\tilde{m})} = \poly(N,1/\eps)$ many $O(q/\eps)$-query $\DDT$s $D_1,\ldots,D_L$, at least one of which, $D_{i^\star}$ is an $\eps$-approximation of $\tilde{R}$.   Since the $L_2$ distance between $R$ and $\tilde{R}$ is $\sqrt{\eps}$, and the $L_2$ distance between $\tilde{R}$ and $D_{i^\star}$ is $\sqrt{\eps}$, by the triangle inequality, the $L_2$ distance between $D_{i^\star}$ and $R$ is at most~$2\sqrt{\eps}$.  Squaring this gives the desired result.  
\end{proof}

\subsection{Step 2: Deterministically identifying a good candidate}

With~\Cref{cor:candidates} in hand, we are now faced with the following task: given an $\RDT$ $R$, a list of $L$ many $\DDT$s $\{ D_1,\ldots,D_L\}$, and the promise that at least one of the $D_i$'s $\eps$-approximates $R$, find one such $\DDT$ deterministically.  This in turn reduces to the task of computing $\| D - \mu_R \|_2^2$ deterministically, which we solve in this subsection. The key idea is to take advantage of the fact that $\RDT$s can be efficiently and deterministically converted into polynomials; specifically, the Fourier representation of~$\mu_R$. 

Let $\mathcal{A}_{\mathrm{Fourier}}$ be the algorithm that takes as input an $\RDT$ $R$ and returns the Fourier representation of $\mu_R$: 
\begin{align*}
    \mathcal{A}_{\mathrm{Fourier}}(R) =  
    \begin{cases} 
       \ \ds \left(\frac{1-x_r}{2}\right)  \mathcal{A}_{\mathrm{Fourier}}(R_0) + \left(\frac{1+x_r}{2} \right) \mathcal{A}_{\mathrm{Fourier}}(R_1)& \text{if $R$'s root queries $x_r$} \vspace{5pt} \\
       \  \frac{1}{2} \cdot (\mathcal{A}_{\mathrm{Fourier}}(R_0) +  \mathcal{A}_{\mathrm{Fourier}}(R_1))& \text{if $R$'s root is a stochastic node} \vspace{5pt} \\
        \  \ell & \text{if $R$ is a leaf $\ell \in [0,1]$,}
    \end{cases} 
\end{align*}
where $R_0$ and $R_1$ are the left and right subtrees of $R$.  It is straightforward to verify by induction that the polynomial $p_R : \{\pm 1\}^n \to [0,1]$, 
\[ p_R(x) = \sum_{S \sse [n]} \wh{p_R}(S) \prod_{i\in S} x_i \] 
returned by $\mathcal{A}_{\mathrm{Fourier}}(R)$ is indeed the Fourier representation of $\mu_R : \zo^n \to [0,1]$: 
\begin{align*}
    p_R(\hat{x}) = \mu_R(x) \quad \text{for all $x\in \zo^n$},
\end{align*}
where $\hat{x}$ denotes that $\{\pm 1\}^n$ representation of $x$. It takes $\poly(N, 2^q)$-time for $\mathcal{A}_{\mathrm{Fourier}}$ to compute all of the nonzero coefficients of the Fourier polynomial representing a $q$-query $\RDT$ with description length $N$. By elementary Fourier analysis, the following two basic properties of $\mu_R$ can be easily ``read off" its Fourier spectrum: 
\begin{align}
   \text{\sl Expectation:}\quad  \E[\mu_R] &= \wh{p_R}(\emptyset) 
    \label{eq:Fourier-expectation}\\
   \text{\sl $2$-norm squared:} \quad \| \mu_R \|_2^2 &= \sum_{S\sse [n]} \wh{p_R}(S)^2. \label{eq:Fourier-variance} 
\end{align}
(The identity (\ref{eq:Fourier-variance}) is commonly known as Parseval's identity.)  The following lemma is now straightforward: 
\begin{lemma}[Deterministic computation of $L_2$ distance]
\label{lemma:compute L2 distance}
    There is a deterministic algorithm with the following guarantee: Given as input a $q_R$-query $\RDT$ $R: \zo^n \times \zo^m \to [0,1]$ and $q_D$-query $\DDT$ $D: \zo^n \to [0,1]$ with description lengths $N_R$ and $N_D$ respectively, it runs in time
    \begin{align*}
        \mathrm{poly}(N_R, N_D, 2^{q_R}, 2^{q_D})
    \end{align*}
    and returns $\| D-\mu_R \|_2^2$. 
\end{lemma}
\begin{proof}
    The algorithm uses $\mathcal{A}_\mathrm{Fourier}$ to compute the Fourier representations for $D$ and $\mu_R$, and then subtracts them to compute the representation for $D - \mu_R$. Then, we use Parseval's identity (\ref{eq:Fourier-variance}) to compute the desired result.
\end{proof}

\Cref{thm:derand-alg} follows from \Cref{cor:candidates} and \Cref{lemma:compute L2 distance}.


\renewcommand{\qbudget}{\overline{q}}

\section{Proof of \Cref{thm:instance-opt,thm:framework}: Instance-optimal derandomization}
\label{sec:instance optimal}

In this section we develop a general framework,~\Cref{thm:framework}, for achieving instance-optimal derandomization.  Our framework will apply to a broad class of error metrics (for measuring the distance between an RDT and a DDT), and we will show that~\Cref{thm:instance-opt} follows as an easy corollary by instantiating this framework with the error metric being $L_2$ error.  Looking ahead, in~\Cref{sec:Nisan} we will show that our instance-optimal constructivization of~\hyperlink{nisan-anchor}{Nisan's Theorem} can also be captured within this framework (though that application requires slightly more technical work). 


The following is the key definition for our framework: 

\begin{definition}[Natural and efficient error metric]
\label{def:natural-efficient} 
We say that an error metric $\mathcal{E} : \{ \RDT\text{s} \} \times \{ \DDT\text{s}\} \to [0,1]$ is \emph{natural} if there is a some $d : [0,1] \times [0,1] \to [0,1]$ such that 
\begin{equation}\mathcal{E}(R,D) = \Ex[d(\mu_R(\bx),D(\bx))]. \label{eq:distance}
\end{equation} 
For $t = t(q)$ a function of $q$, we say that $\mathcal{E}$ is \emph{$t$-efficient} if for all $q$-query $\RDT$s $R$ and $\DDT$s $D$ of description lengths $N_R$ and $N_D$ respectively, 
\begin{enumerate}
\item There is a deterministic $\poly(t,N_R,N_D)$-time algorithm that computes $\mathcal{E}(R,D)$. 
\item There is a deterministic $\poly(t,N_R)$-time algorithm that computes the constant $c \in [0,1]$ that minimizes $\mathcal{E}(R,c)$. 
\end{enumerate}  
\end{definition}



\begin{theorem}[General framework for instance-optimal derandomization]
    \label{thm:framework-formal}
    Let $\mathcal{E}$ be a natural $t$-efficient error metric. There is a deterministic algorithm, $\mathcal{A}_{\mathrm{InstanceOpt},\mathcal{E}}$ with the following guarantee: Given as input a $q$-query $\RDT$ $R$ with description length $N$ and an error parameter $\eps \in (0,1)$, for
    \[ q^\star_{R,\mathcal{E}} \coloneqq \min \{\,q' \colon \text{there is a $q'$-query $\DDT$ $D$ such that $
    \mathcal{E}(R,D) \leq \epsilon$\,}\}. \] 
    $\mathcal{A}_{\mathrm{InstanceOpt},\mathcal{E}}$ runs in
    \[ \poly(N,t, n^{q^{\star}_{R,\mathcal{E}}})  \] 
    time and returns a $q^{\star}_{R,\mathcal{E}}$-query $\DDT$ $D$ satisfying $\mathcal{E}(R,D) \leq \epsilon$.  
\end{theorem}


The algorithmic core of \Cref{thm:framework-formal} is the deterministic recursive backtracking procedure {\sc Find} shown in \Cref{fig:find}, the goal of which is to finds a $\qbudget$-query decision tree that achieves minimal error relative to a given error metric~$\mathcal{E}$.

The assumptions that $\mathcal{E}$ is natural and $t$-efficient will both play crucial roles in our analysis of {\sc Find}: the former is the key criterion for establishing its correctness (\Cref{lemma:find correctness}), and the latter is the key criterion for analyzing its runtime (\Cref{lemma:find efficiency}). 
\begin{figure}[h!]
  \captionsetup{width=.9\linewidth}
\begin{tcolorbox}[colback = white,arc=1mm, boxrule=0.25mm]
    \textsc{Find}$(R, \mathcal{E}, \qbudget, \pi)$:
    \begin{itemize}[align=left]
        \item[\textbf{Input:}]  An $\RDT$ $R$, an error metric $\mathcal{E}$, query budget $\qbudget$, and restriction $\pi$.
        \item[\textbf{Output:}]  A $\qbudget$-query $\DDT$ $D$ that minimizes $\mathcal{E}(R_\pi,D)$ among all $\qbudget$-query $\DDT$s.
    \end{itemize}
    \begin{enumerate}
        \item If $\qbudget = 0$, return the constant $c\in [0,1]$ that minimizes $\mathcal{E}(R_\pi, c)$.
        \item For every $i \in [n]$, let $D_i$ be the $\DDT$ defined as follows: 
        \begin{itemize}
        \item[$\circ$] $D_i$ queries $x_i$ at the root; 
        \item[$\circ$] Has $\textsc{Find}(R, \mathcal{E}, \qbudget-1, \pi \cup \{x_i \leftarrow 0\})$ as its left subtree;
        \item[$\circ$] Has $\textsc{Find}(R, \mathcal{E}, \qbudget-1, \pi \cup \{x_i \leftarrow 1\})$ as its right subtree.
        \end{itemize}
Here $\pi \cup \{ x_i \leftarrow b\}$ denotes the extension of $\pi$ where $x_i$ is set to $b$. 
        \item Return the tree $D_{i^\star}$ that minimizes $\mathcal{E}(R_\pi, D_{i^\star})$ among all $i^\star \in [n]$.
    \end{enumerate}
\end{tcolorbox} 
\caption{A deterministic recursive backtracking algorithm for finding a $\qbudget$-query $\DDT$ of minimal error relative to an error metric $\mathcal{E}$.} 
\label{fig:find}
\end{figure}

\begin{lemma}[Correctness of \textsc{Find}]
    \label{lemma:find correctness}
    Consider any $\RDT$ $R$, natural error metric $\mathcal{E}$, query budget $\qbudget \in \Z_+$, and restriction $\pi$.   The algorithm $\textsc{Find}(R, \mathcal{E}, \qbudget, \pi)$ of~\Cref{fig:find} returns a $\qbudget$-query $\DDT$ $D$ that minimizes $\mathcal{E}(R_\pi,D)$ among all $\qbudget$-query $\DDT$s.
\end{lemma}
\begin{proof}
    We proceed by induction on $\qbudget$. If $\qbudget = 0$, then $\textsc{Find}$ returns at Step 1 and is clearly correct.
    For the inductive step, suppose that $\qbudget \geq 1$. For any $i\in [n]$, we first claim that the tree $D_i$ defined in Step 2 is a $\qbudget$-query $\DDT$ for $R_\pi$ that achieves minimal error among those that query $x_i$ at the root. Let $(D_i)_{\mathrm{left}}$ and $(D_i')_{\mathrm{right}}$ be its left and right subtrees respectively.  Now our assumption that $\mathcal{E}$ is a natural error metric, we have that: 
    \begin{align*}
        \mathcal{E}(R_\pi, D_i) = \lfrac{1}{2}\big(\mathcal{E}(R_{\pi \cup \{x_i\leftarrow 0\}}, (D_i)_{\mathrm{left}}) + \mathcal{E}(R_{\pi \cup \{x_i\leftarrow 1\}}, (D_i)_{\mathrm{right}})\big).
    \end{align*}
    By the inductive hypothesis, the left and right subtrees $(D_i)_{\mathrm{left}}$ and $(D_i)_{\mathrm{right}}$ are $(\qbudget-1)$-query $\DDT$s that have minimal error with respect to $R_{\pi \cup \{x_i \leftarrow 0\}}$ and $R_{\pi \cup \{x_i\leftarrow 1\}}$ respectively. Hence indeed, $D_i$ is a $\qbudget$-query $\DDT$ for $R_\pi$ that achieves a minimal error among those that query $x_i$ at the root.  
    
    Since $\textsc{Find}$ returns the $D_{i^\star}$ that minimizes $\mathcal{E}(R_\pi, D_{i^\star})$ among all $i^\star \in [n]$ in Step 3, and each $D_i$ is $\qbudget$-query $\DDT$ for $R_\pi$ that achieves minimal error tree among those that query $x_i$ at the root, we conclude that $\textsc{Find}$ returns a minimal error tree among all $\qbudget$-query $\DDT$s.
\end{proof}

\begin{lemma}[Efficiency of {\sc Find}]
\label{lemma:find efficiency} 
    Consider any $q$-query $\RDT$ $R$ with description length $N$, error function $\mathcal{E}$ that is $t$-efficient, $\qbudget \in \Z_+$, and restriction $\pi$.   The algorithm $\textsc{Find}(R, \mathcal{E}, \qbudget, \pi)$ of~\Cref{fig:find} takes time $\poly(N,t,n^{\qbudget})$.   
\end{lemma}
\begin{proof}
   Let $T(\qbudget)$ denote the running time of {\sc Find} when run with query budget $\qbudget$.   If $\qbudget = 0$ then the algorithm only executes Step 1, which we claim can be done in time $\poly(N,t)$. In time $\poly(N)$ we can convert $R$ to $R_\pi$ by skipping any decision nodes restricted by $\pi$ and replacing them with the subtree on the side specified by $\pi$. Then, since $\mathcal{E}$ is $t$-efficient, we can compute the constant $c\in [0,1]$ that minimizes $\mathcal{E}(R_\pi, c)$ in time $\poly(N,t)$.
    
Next we consider the case of $\qbudget \ge 1$.  In step 2, $\textsc{Find}$ recurses $2n$ times, each with $\qbudget$ decremented by one. By induction, all of these recursive calls takes total time $2n \cdot T(\qbudget-1).$  In step 3, $\textsc{Find}$ must compute $\mathcal{E}(R_\pi, D_i)$ for up to $n$ different coordinates $i$, where each $D_i$ has size at most $2^{\qbudget}$. Since $\mathcal{E}$ is $t$-efficient, this takes time at most $n \cdot \poly(N,t,2^{\qbudget}).$  We therefore have the recurrence relation: 
\[ T(\qbudget) \le 2n \cdot T(\qbudget -1) + n\cdot \poly(N,t,2^{\qbudget}).\] 
 Solving this recurrence relation gives us the claimed bound $T(\qbudget) \le \poly(N, t,n^{\qbudget}).$
\end{proof}

Now that we have proved the correctness and runtime of $\textsc{Find}$, we show how to use it in our framework for instance-optimal derandomization: 
\begin{proof}[Proof of~\Cref{thm:framework-formal}]
    Let $\mathcal{A}_{\mathrm{InstanceOpt},\mathcal{E}}$ be the algorithm that runs
    \begin{align*}
        \textsc{Find}(R, \mathcal{E}, \qbudget = j, \pi = \emptyset)
    \end{align*} 
    for $j = 0,1,2,\ldots$ and returns the first output of $\textsc{Find}$ that has error at most $\varepsilon$ relative to $\mathcal{E}$. By \Cref{lemma:find correctness}, $\mathcal{A}_{\mathrm{InstanceOpt},\mathcal{E}}$ will return a $q^{\star}_{R,\mathcal{E}}$-query $\DDT$ $D$ satisfying $\mathcal{E}(R,D) \leq \epsilon$.  By~\Cref{lemma:find efficiency}, the runtime of $\mathcal{A}_{\mathrm{InstanceOpt},\mathcal{E}}$ is     \begin{align*}
        \sum_{j=0}^{q^\star_{R,\mathcal{E}}} \poly(N,t,n^{j}) \le  \poly(N,t,n^{q^\star_{R,\mathcal{E}}}).
    \end{align*}
 This completes the proof of~\Cref{thm:framework-formal}. 
\end{proof}

\subsection{Using this framework to prove~\Cref{thm:instance-opt}: $L_2$ error is natural and efficient}

In order to apply our general framework,\Cref{thm:framework-formal}, we need to show that squared $L_2$ error is natural and efficient, as defined in \Cref{def:natural-efficient}. Clearly, it is natural for $d(x,y) = d(x-y)^2$. The following Lemma, combined with \Cref{lemma:compute L2 distance}, shows it is efficient.

\begin{lemma}
    \label{lemma:min l2 distance}
    There is a deterministic algorithm with the following guarantee: Given as input a $q$-query $\RDT$ $R: \zo^n \times \zo^m \to [0,1]$ with description length $N$, it runs in time
    \begin{align*}
        \mathrm{poly}(N, 2^{q})
    \end{align*}
    and finds the constant $c \in [0,1]$ that minimizes $\E[(c-\mu_R(\bx))^2]$. 
\end{lemma}
\begin{proof}
The quantity $\Ex_{\bx}[(c - \mu_R(\bx))^2]$ is a convex function of $c$ with derivative, with respect to $c$, of the following expression.
\begin{align*}
    \E\big[2(c - \mu_R(\bx))\big].
\end{align*}
This is equal to $0$ only when $c = \E[\mu_R(\bx)]$, which is the unique minimum of $\E[(c-\mu_R(\bx))^2]$. To find it, we use $\mathcal{A}_{\mathrm{Fourier}}$ to convert $R$ to a polynomial and then use (\ref{eq:Fourier-expectation}) to compute the optimal $c$. This takes time $\mathrm{poly}(N, 2^{q})$.
\end{proof}

Since $L_2$ error is natural and efficient, \Cref{thm:instance-opt} is a consequence of our general framework, \Cref{thm:framework-formal}.

\subsection{Extensions and variants of our framework} 
\label{sec:extensions} 
The framework of~\Cref{thm:framework-formal} seems fairly versatile and amendable to variants; we will rely on this versatility for a couple of applications in this work: 

\begin{enumerate}
\item {\sl $\zo$-valued DTs and constructivizing~\hyperlink{nisan-anchor}{Nisan's Theorem}:}  In order to apply~\Cref{thm:framework-formal} to constructivize~\hyperlink{nisan-anchor}{Nisan's Theorem}, which concerns query algorithms for {\sl decision} problems, we will need to specialize it to $\zo$-valued $\RDT$s and $\DDT$s.  In this context, an error function $\mathcal{E}$ is {\sl natural} if the condition (\ref{eq:distance}) holds for some $d : \zo \times \zo \to [0,1]$ and $\zo$-valued $R$ and $D$. Similarly, it is {\sl $t$-efficient} if there are corresponding deterministic algorithms for $\zo$-valued $\RDT$s and $\DDT$s that satisfy the requirements of \Cref{def:natural-efficient}.
\item {\sl Instance-optimal DTs for polynomials and constructivizing~\hyperlink{beals-anchor}{Beals et al.'s Theorem}:}  In order to apply~\Cref{thm:framework-formal} to constructivize~\hyperlink{beals-anchor}{Beals et al.'s Theorem} (\Cref{thm:beals}), which concern {\sl quantum} query algorithms, we will need the following generalization of it: while~\Cref{thm:framework-formal} as gives an algorithm for finding an instance-optimal $\DDT$ for a $q$-query $\RDT$, it can in fact be used to find an instance-optimal $\DDT$ for an arbitrary degree-$q$ {\sl polynomial} $p : \zo^n \to [0,1]$ (again with respect to an error metric $\mathcal{E}$).\footnote{To see the relationship between $\RDT$s and polynomials, note that if $R$ is a $q$-query $\RDT$ then $\mu_R$ is a degree-$q$ polynomial.}  For this generalization, one just has to make the corresponding adjustments to~\Cref{def:natural-efficient} (natural and efficient error metrics), so that $\mathcal{E}$ now measures the distance between an arbitrary degree-$q$ polynomial and a $\DDT$. 

\item {\sl Beyond the uniform distribution.}  While we have stated~\Cref{def:natural-efficient} so that $\mathcal{E}$ is defined with respect to a uniform random $\bx \sim \zo^n$,~\Cref{thm:framework-formal} in fact applies to all other distributions.  (We do not explore this generalization in this work.) 
\end{enumerate} 





\section{Proof of~\Cref{thm:online}: Online derandomization}

In this section we will prove \Cref{thm:online}.  We will actually prove the following ``high probability version" of~\Cref{thm:online}, which yields~\Cref{thm:online} (the ``expectation version") as an immediate corollary:  

\begin{theorem}[Online derandomization] 
\label{thm:online-high-probability} 
There is a deterministic algorithm $\mathcal{A}_{\mathrm{Online}}$ with the following guarantee.  Given as input a randomized $q$-query algorithm $R : \zo^n \times \zo^m \to [0,1]$ with description length $N$, an error parameter $\eps \in (0,\frac1{2})$, and an input $x \in \zo^n$, this algorithm $\mathcal{A}_{\mathrm{Online}}$ runs in 
\[ \poly(N, q,1/\eps) \] 
 time, makes $O(q^2/\eps^2\delta^2)$ queries to $x$,   and returns a value $\mathcal{A}_{\mathrm{Online}}(x) \in [0,1]$.  The output values of $\mathcal{A}_{\mathrm{Online}}$ satisfy:  
 \[ \Prx_{\bx}\Big[ \big|\mathcal{A}_{\mathrm{Online}}(\bx) - \Ex_{\br} [R(\bx, \br)]\big| \ge \eps \Big] \le \delta. \] 
\end{theorem}

In~\Cref{sec:OSSS} we prove a key new structural fact, a generalization of the O'Donnell, Saks, Schramm, Servedio inequality~\cite{OSSS05} from deterministic to randomized decision trees.  In~\Cref{sec:AA-analysis}, we use this structural fact to prove~\Cref{thm:online-high-probability}.


\subsection{Every randomized DT has an influential variable} 
\label{sec:OSSS}

We need a few basic definitions in order to state the new structural fact that we prove.  

\begin{definition}[Probability of querying a coordinate]
Let $D$ be a $\DDT$.  For each $i\in [n]$, we define $\delta_i(D)$ to be the probability that $D$ queries $\bx_i$ where $\bx \sim \zo^n$ is a uniform random input.  For an $\RDT$ $R : \zo^n \times \zo^m \to [0,1]$, we define $\delta_i(R)$ analogously: 
\[ \delta_i(R) \coloneqq \Ex_{\br}[\delta_i(R_{\br})].\] 
\end{definition}

\begin{definition}[Influence of variables] 
Let $f : \zo^n \to [0,1]$.  For each $i\in [n]$, we define the \emph{influence of variable $i$ on $f$} to be the quantity 
\[ \Inf_i(f) \coloneqq \Ex_{\bx}[|f(\bx)-f(\bx^{\oplus i})|], \]
where $\bx^{\oplus i}$ denotes $\bx$ with its $i$-th coordinate flipped.  The \emph{total influence} of $f$ is $\Inf(f) \coloneqq \sum_{i=1}^n \Inf_i(f)$. 
\end{definition}

The following powerful inequality from the analysis of boolean functions is due to O'Donnell, Saks, Schramm, and Servedio~\cite{OSSS05}.  It relates the influences of variables to query complexity: 

\begin{theorem}[\cite{OSSS05} inequality: Every $\DDT$ has an influential variable]  
Let $D : \zo^n \to \zo$ be a $q$-query $\DDT$.  Then 
\[ \Var[D] \le \sum_{i=1}^n \delta_i(D) \cdot \Inf_i(D).  \] 
Consequently, there must exist an $i^\star\in [n]$ such that 
\[ \Inf_{i^\star}(D) \ge \frac{\Var[D]}{\Delta(D)} \ge \frac{\Var[D]}{q},\]
where $\Delta(D) \coloneqq \sum_{i=1}^n \delta_i(D)$ is the average depth of $D$. 
\end{theorem} 

Our first main result in this subsection,~\Cref{thm:OSSS-RDT}, is a generalization of the~\cite{OSSS05} inequality from to $\DDT$s to $\RDT$s.  We will show that this generalization follows from a {\sl different} generalization of their inequality, the ``two-function version" of the~\cite{OSSS05} inequality.  

The following is a special case of Theorem 3.2 of~\cite{OSSS05} (see the discussion right before their Section 3.4), rewritten in notation that will be especially convenient for us:

\begin{theorem}[Two-function version of OSSS]
\label{thm:OSSS}
Let $D : \zo^n \to [0,1]$ be a $\DDT$ and $\mu : \zo^n \to [0,1]$ be any function.\footnote{A remarkable feature of~\Cref{thm:OSSS} is that $D$ and $\mu$ can be two arbitrary functions, completely unrelated to each other.}  Then 
\[ |\Cov[D,\mu]| \le \sum_{i=1}^n \delta_i(D) \cdot \Inf_i(\mu). \]
\end{theorem}

We now derive the following as a corollary of~\Cref{thm:OSSS}: 

\begin{theorem}[Every $\RDT$ has an influential variable]
\label{thm:OSSS-RDT}
Let $R : \zo^n \times \zo^m \to [0,1]$ be a $q$-query $\RDT$ and $\mu_R : \zo^n \to [0,1]$ be its mean function.  Then 
\[ \Var[\mu_R] \le \sum_{i=1}^n \delta_i(R)\cdot \Inf_i(\mu_R). \] 
Consequently, there must exist an $i^\star \in [n]$ such that 
\[ \Inf_{i^\star}(\mu_R) \ge \frac{\Var[\mu_R]}{\Delta(R)} \ge \frac{\Var[\mu_R]}{q}.\] 
where $\Delta(R) \coloneqq \sum_{i=1}^n \delta_i(R)$. 
\end{theorem}

\begin{proof} 
For clarity, we drop the subscript on $\mu_R$.  Viewing $R$ as a distribution over $q$-query $\DDT$s $R_{\br}$ for $r\sim \zo^m$, we begin by applying~\Cref{thm:OSSS} to each $\DDT$ in the support of $R$: 
\[ \Ex_{\br}[|\Cov[R_{\br},\mu]|] \le \Ex_{\br} \Bigg[ \sum_{i=1}^n \delta_i(R_{\br}) \cdot \Inf_i(\mu)\Bigg] =  \sum_{i=1}^n \delta_i(R) \cdot \Inf_i(\mu).\] 
Rewriting the LHS of the above, 
\begin{align*}
\Ex_{\br}[|\Cov[R_{\br},\mu]|]  &\ge \big| \Ex_{\br}[\Cov[R_{\br},\mu]] \big| \tag*{($\E[|\bX|] \ge |\E[\bX]|$ for all r.v.'s $\bX$)}\\
&= \bigg|\Ex_{\br}\Big[\Ex_{\bx}[(R_{\br}(\bx)-\Ex_{\bx}[R_{\br}(\bx)])(\mu(\bx)-\Ex_{\bx}[\mu(\bx)])]\Big]\bigg| \tag*{(Definition of covariance)} \\
&= \bigg|\Ex_{\bx}\Big[\Ex_{\br}[(R_{\br}(\bx)-\Ex_{\bx}[R_{\br}(\bx)])(\mu(\bx)-\Ex_{\bx}[\mu(\bx)])]\Big]\bigg| \tag*{(Swapping expecations)} \\
&= \bigg|\Ex_{\bx}\Big[(\mu(\bx)-\Ex_{\bx}[\mu(\bx)])(\mu(\bx)-\Ex_{\bx}[\mu(\bx)])\Big]\bigg| \tag*{(Definition of $\mu$)} \\
&=  \Ex_{\bx}\Big[(\mu(\bx)-\Ex_{\bx}[\mu(\bx)])^2\Big]   \\
&= \Var(\mu).
\end{align*}
This completes the proof of~\Cref{thm:OSSS-RDT}. 
\end{proof} 

\begin{remark}[Other known extensions of the~\cite{OSSS05} inequality]
In~\cite{OSSS05} the authors show that their inequality extend to randomized decision trees that compute functions $f: \zo^n\to\zo$ with {\sl zero error}.  In our notation, these are functions $R : \zo^n \times \zo^m \to \zo$ that are promised to satisfy $\mu_R(x) = f(x)$ for all $x\in \zo^n$ (cf.~\Cref{def:decision-RDTs}). 

\violet{For $\RDT$s $R$ that compute functions $f : \zo^n \to \zo$ with $\eps$ error, Jain and Zhang~\cite{JZ11} proved the following variant of the~\cite{OSSS05} inequality:
\[ \min\{\Pr[f(\bx)=1],\Pr[f(\bx)=0]\} - \eps \le \sum_{i=1}^n \delta_i(R) \cdot \Inf_i(f). \] 
This does not apply to general $\RDT$s where no assumptions are made about the distribution of output values of $R$ on a given input $x$ (in particular, where $\mu_R(x)$ is not assumed to be close to $0$ or $1$). 

To our knowledge, our extension of the~\cite{OSSS05} inequality to general $\RDT$s,~\Cref{thm:OSSS-RDT}, was not known previously known (though as we just showed, it is a fairly straightforward consequence of the two-function version generalization of the~\cite{OSSS05} inequality). }
\end{remark}

\vspace{-5pt} 

\pparagraph{Total influence of RDTs.}   We complement~\Cref{thm:OSSS-RDT} with an upper bound on the {\sl total} influence of $\RDT$s.  The following is a basic fact in concrete complexity and is easy to verify: 

\begin{fact}[Total influence of $\DDT$s] 
\label{fact:DT-influence} 
Let $D : \zo^n \to [0,1]$ be a $q$-query $\DDT$.  Then $\Inf(D) \le q$. 
\end{fact}

We will need the following generalization of~\Cref{fact:DT-influence} from $\DDT$s to $\RDT$s: 

\begin{corollary}[Total influence of $\RDT$s]
\label{cor: total influence}
Let $R : \zo^n \times \zo^m \to [0,1]$ be a $q$-query $\RDT$ and $\mu_R : \zo^n \to [0,1]$ be its mean function. 
Then $\Inf(\mu_R) \le q$. 
\end{corollary} 

\begin{proof}
Again, for clarity we drop the subscript on $\mu_R$.  We have that: 
\begin{align*}
\Inf(\mu) &= \sum_{i=1}^n  \Ex_{\bx}[|\mu(\bx) - \mu(\bx^{\oplus i})|]  \tag*{(Definition of total influence)} \\
&= \sum_{i=1}^n \Ex_{\bx}\Big[ \big| \Ex_{\br}[R_{\br}(\bx)] - \Ex_{\br}[R_{\br}(\bx^{\oplus i})]  \big|\Big]  \tag*{(Definition of $\mu$)}\\
&= \sum_{i=1}^n \Ex_{\bx}\Big[ \big| \Ex_{\br}[R_{\br}(\bx)- R_{\br}(\bx^{\oplus i})] \big| \Big] \\
&\le \sum_{i=1}^n \Ex_{\bx} \Big[ \Ex_{\br} [ | R_{\br}(\bx) - R_{\br}(\bx^{\oplus i})|] \Big]  \tag*{($|\E[\bX]|\le \E[|\bX|]$ for all r.v.'s $\bX$)} \\
&= \sum_{i=1}^n \Ex_{{\br}} \Big[ \Ex_{\bx} [ | R_{\br}(\bx) - R_{\br}(\bx^{\oplus i})|] \Big]  \tag*{(Swapping expectations)} \\
&= \Ex_{{\br}} \Bigg[  \sum_{i=1}^n \Ex_{\bx} [ | R_{\br}(\bx) - R_{\br}(\bx^{\oplus i})|]\Bigg] \le q, 
\end{align*}
where the final inequality holds by applying~\Cref{fact:DT-influence} to each $R_{\br}$. 
\end{proof}

\subsection{Most-influential-at-the-root algorithm}
\label{sec:AA-analysis} 

We will first show an algorithm for building a deterministic decision tree $D$ that approximates a randomized decision tree $R$ by iteratively querying the most influential variable of $\mu_R$. This is {\sl not} the online algorithm described in \Cref{thm:online-high-probability}, but due to the ``top-down" fashion in which it constructs $D$, it can be easily modified to yield an online variant.  Indeed, the actual algorithm of~\Cref{thm:online-high-probability} and its analysis will follow very easily from our analysis of this algorithm.

\begin{lemma}
\label{lem:most influential}
Let $D$ be the $(q^2/\eps^2\delta^2)$-query deterministic algorithm returned by the algorithm {\sc BuildTopDownDT}($R$, $\eps$, $\delta$) described in~\Cref{fig:most-influential}.  Then 
\[ \Pr[|D(\bx) -\mu_R(\bx)| \ge \eps] \le 2 \delta.\] 
\end{lemma} 

\begin{figure}[H]
  \captionsetup{width=.9\linewidth}
\begin{tcolorbox}[colback = white,arc=1mm, boxrule=0.25mm]
\vspace{3pt} 
{\sc BuildTopDownDT}($R$, $\eps$, $\delta$):  \vspace{6pt} 

\ \ Let $\mu = \mu_R$ denote the mean function of $R$, and initialize $D$ to be the empty tree. \vspace{4pt} 

\ \ for $d = 0,  \ldots, q^2/\eps^2\delta^2$: 
\vspace{-3pt}

\begin{itemize}

\item[] {\sl Query most influential variable}: For each of the $2^{\ell}$ leaves $\ell$ in $D$, let $x_{i(\ell)}$ denote the most influential variable of the subfunction $\mu_\ell$ of $\mu$: 
\[ \Inf_{i(\ell)}(\mu_\ell) \ge \Inf_j(\mu_\ell) \quad \text{for all $j\in [n]$.} \] 
Grow $D$ by replacing $\ell$ with a query to $x_{i(\ell)}$. 
\end{itemize}

\ \ for each leaf $\ell$ of $D$: 
\vspace{-3pt}
\begin{itemize}
\item[] Assign $\ell$ the value $\E[\mu_\ell]$. 
\end{itemize}
\end{tcolorbox}
\caption{Most-influential-at-the-root algorithm}
\label{fig:most-influential} 
\end{figure}

\begin{proof}
We define the {\sl average subfunction influence at depth $d$ of $D$} to be: 
\[ \AvgInf_d(D) \coloneqq \mathop{\Ex_{\text{paths $\bpi$ in $D$}}}_{|\bpi|=d}[\Inf(\mu_\pi)], \] 
where the expectation is taken over a random path $\bpi$ from the root of $D$ to a node at depth $d$.  
The proof proceeds via a potential function argument, using average subfunction influence as our progress measure.  We will need a simple observation: for all functions $f : \zo^n \to [0,1]$ and coordinates $i\in [n]$,
\begin{equation} \Inf(\mu) = \Inf_i(f) + \lfrac1{2}(\Inf(f_{x_i=0}) + \Inf(f_{x_i=1})). \label{eq:query-most-inf}
\end{equation} 
Writing $x(\bpi)$ to denote that variable queried at the end of $\bpi$ in $D$ (equivalently, the variable queried at the root of $D_{\bpi}$),  we have that: 
\begin{align*}
\AvgInf_{d+1}(D) &\le \AvgInf_{d}(D) - \mathop{\Ex_{\text{paths $\bpi$ in $D$}}}_{|\bpi|=d}[\Inf_{x(\bpi)}(\mu_{\bpi})]  \tag*{(Equation (\ref{eq:query-most-inf}))} \\
&\le \AvgInf_{d}(D) - \mathop{\Ex_{\text{paths $\bpi$ in $D$}}}_{|\bpi|=d}\Bigg[ \frac{\Var(\mu_{\bpi})}{q}\Bigg]. \tag*{(\Cref{thm:OSSS-RDT})} 
\end{align*} 
%
%
%
%
%

At each depth $d$, we must have one of two cases: either the following equation holds, or it does not. 

\begin{equation}
\label{eq:variance bound}
\mathop{\Ex_{\text{paths $\bpi$ in $D$}}}_{|\bpi|=d}[\Var(\mu_{\bpi})] < \eps^2\delta^2
\end{equation}

\begin{enumerate}
    \item  (\Cref{eq:variance bound} holds):  By Markov's inequality, we have $$\Pr[\Var(\mu_{\bpi}) \geq \eps^2\delta] \leq \frac{\Ex[\Var(\mu_{\bpi})]}{\eps^2\delta} \le \delta.$$ 
    
For the $(1 -\delta)$-fraction of paths $\pi$ that satisfy $\Var(\mu_\pi) \leq \eps^2\delta$, we apply Chebyshev's inequality to get: 
   \[ \Pr[|\mu_\pi - \E[\mu_\pi]| \geq \eps] \leq \delta.\]
    
    \item (\Cref{eq:variance bound} does not hold):  By \Cref{eq:query-most-inf}, we have the following: 
\[ \AvgInf_{d+1}(D) \leq \AvgInf_d(D) -  \frac{\eps^2 \delta^2}{q}.\] 
    
\end{enumerate}

The following is a consequence of the law of total variance:

\[ \mathop{\Ex_{\text{paths $\bpi$ in $D$}}}_{|\bpi|=d + 1}[\Var(\mu_\pi)]\le \mathop{\Ex_{\text{paths $\bpi$ in $D$}}}_{|\bpi|=d}[\Var(\mu_\pi)].\] 

Therefore, if there is some depth $d^*$ for which Case 1 applies, then Case 1 continues to apply for all $d \ge d^*$. By \Cref{cor: total influence}, we know that the total influence $\Inf(\mu_R) \leq q$, and so we start with $\AvgInf_0(D) = \Inf(\mu_R) \leq q$. Since average influence is a non-negative quantity, we can have Case 2 for only $\leq q^2/\eps^2\delta^2$ depths before we reach a $d^*$ which is in Case 1. 
The lemma follows by running {\sc BuildTopDownDT} for $q^2/\eps^2\delta^2 + 1$ levels.
\end{proof}

\subsection{Deterministic quadratic-time algorithm for computing influence}

\begin{lemma}[Algorithm for computing influence]
Given a description of an $\RDT$ $R$ with description length $N$, for any $i\in [n]$ the influence of variable $i$ on $\mu_R$, 
\[ \Inf_i(\mu_R) \coloneqq \E[|\mu(\bx)-\mu(\bx^{\oplus i})|] \] 
can be computed deterministically in time $O(N^2)$.
\end{lemma}

\begin{proof}
We write $\mu$ for $\mu_R$. We first consider the simpler problem of deterministically computing the influence of the variable queried at the root of $R$.  Suppose that $x_i$ is queried at the root of $R$.  Let $R_{\mathrm{left}}$ and $R_{\mathrm{right}}$ denote the left and right subtrees of $R$, and $\mu_{\mathrm{left}}$ and $\mu_{\mathrm{right}}$ be their mean functions.  In this case, we have that
\begin{align*}
\Inf_i(\mu) &= \Ex_{\bx}[|\mu(\bx) - \mu(\bx^{\oplus i})|] \\
&= \Ex_{\bx}[|\mu_{\mathrm{left}}(\bx) - \mu_{\mathrm{right}}(\bx)|]  \\
&= \mathop{\sum_{\mathrm{paths}}}_{\pi \in R_{\mathrm{left}}} \Prx_{\bx}[\,\text{$\bx$ follows $\pi$}\,] \cdot\Ex_{\bx}[|(R_{\mathrm{right}})_\pi(\bx)-\ell(\pi)|]  \\
&= \mathop{\sum_{\mathrm{paths}}}_{\pi \in R_{\mathrm{left}}} 2^{-|\pi|} \cdot\Ex_{\bx}[|(R_{\mathrm{right}})_\pi(\bx)-\ell(\pi)|],
\end{align*} 
where $\ell(\pi)$ denotes the value of leaf at the end of path $\pi$.  This quantity can be computed deterministically using the algorithm given in~\Cref{fig:TopDown}.

\begin{figure}[H]
  \captionsetup{width=.9\linewidth}
\begin{tcolorbox}[colback = white,arc=1mm, boxrule=0.25mm]
\vspace{3pt} 

{\sc RootInfluence}($R$):  \vspace{6pt} 

\ \ Inititialize $\Inf_{\mathrm{root}}$ to 0. \vspace{4pt} 

\ \ for each path $\pi$ in $R_{\mathrm{left}}$: 
\vspace{-3pt} 
\begin{enumerate}
\item {\sl Restrict $R_{\mathrm{right}}$ by $\pi$:} Compute $(R_{\mathrm{right}})_\pi$ as follows: for each decision node $x_j$ restricted by $\pi$, replace every occurrence of $x_j$ in $R_{\mathrm{right}}$ by its subtree on the side specified by $\pi$.   

\item {\sl Path counting:} Let $\ell(\pi)$ be the value of the leaf at the end of $\pi$. Compute $p = \E[|(R_{\mathrm{right}})_\pi(\bx) - \ell(\pi)|]$ as follows: \vspace{4pt} 

\ \ Initialize $p$ to 0. 

\ \ for each path $\sigma$ in $(R_{\mathrm{right}})_\pi$: 
\vspace{-3pt} 
\begin{enumerate}
     \item[] Increment $p = p+2^{-|\sigma|} \cdot |\ell(\sigma)-\ell(\pi)|$. 
\end{enumerate}

\item {\sl Update:} Increment $\Inf_{\mathrm{root}} = \Inf_{\mathrm{root}} + p \cdot 2^{-|\pi|}$. 
\end{enumerate}

\ \ Output: $\Inf_{\mathrm{root}}$. 
\end{tcolorbox}
\caption{Deterministic algorithm to compute the influence of the root of an $\RDT$.}
\label{fig:TopDown}
\end{figure}

%
%
%
%

Since $R_{\mathrm{left}}$ and $R_{\mathrm{right}}$ each have at most $N$ paths, the total runtime of {\sc RootInfluence} is $O(N^2)$.   With {\sc RootInfluence} in hand, the influence of a variable that is not queried at the root of $R$ is easy to compute.  First note that: 
\begin{align*}
 \Inf_i(\mu) &= \mathop{\sum_{\text{Subtrees $T$}}}_{\text{rooted at $x_i$}} \Pr[\,\text{$\bx$ visits $T$}\,] \cdot \textsc{RootInfluence}(T) \\
 &= \mathop{\sum_{\text{Subtrees $T$}}}_{\text{rooted at $x_i$}} 2^{-\mathrm{depth}(T,R)} \cdot \textsc{RootInfluence}(T),
 \end{align*} 
where $\mathrm{depth}(T,R)$ is the depth of the root of $T$ (which queries $x_i$) within $R$.  Therefore, we can compute $\Inf_i(\mu)$ simply by calling {\sc RootInfluence} on each subtree rooted at each occurrence of $x_i$ in $R$. The sum of sizes of these subtrees is at most $N$. Since $a^2 + b^2 \leq (a + b)^2$ for any positive $a$ and $b$, the sum of the runtimes of {\sc RootInfluence} on these subtrees is $O(N^2)$ as well. 
\end{proof}

\subsection{Efficient computation of paths} 

We now show that \Cref{thm:online-high-probability} follows from the following algorithm.

\begin{figure}[H]
  \captionsetup{width=.9\linewidth}
\begin{tcolorbox}[colback = white,arc=1mm, boxrule=0.25mm]
\vspace{3pt} 

{\sc BuildTopDownPath}($\underline{x}$, $R$, $\eps$, $\delta$):  \vspace{6pt} 

\ \ Let $\mu = \mu_R$ denote the mean function of $R$, and initialize $\pi$ to be the empty path. \vspace{4pt} 

\ \ for $d = 0,  \ldots, q^2/\eps^2\delta^2$: 
\vspace{-3pt} 
\begin{enumerate}
\item {\sl Compute influences:}  Compute the variable influences of $\mu_\pi$, and let $i^\star$ be the most influential variable. 
\item Extend $\pi$ by restricting the $i^\star$-th coordinate to $\underline{x}_{i^\star}$. 
\end{enumerate}
\ \ Output $\E[\mu_{\pi}]$. 
\end{tcolorbox}
\caption{Deterministic online algorithm for approximating $\mu_R(\underline{x})$.}
\label{fig:TopDown}
\end{figure}

\begin{proof}

The correctness and accuracy guarantees of this algorithm follow directly from \Cref{lem:most influential}. The algorithm runs for $q^2/\eps^2\delta^2$ iterations, computing variable influences on each iteration, for each variable which appears in $R$. Computing all relevant influences takes $O(nN^2) = O(N^3)$ time. Thus the full algorithm takes $\poly(N, q, 1/\eps, 1/\delta)$ time, which concludes the proof of \Cref{thm:online-high-probability}. 
\end{proof}

\begin{remark}
We observe that {\sc BuildTopDownPath} is also highly memory efficient.  It uses only $O(q+m)$ space: this is the maximum number of bits that may be needed to store the influence of a variable in a $q$-query RDT with randomness complexity $m$. 
\end{remark}


\newcommand{\BayesError}{\mathrm{BayesError}}
\section{Constructivization of Nisan's Theorem}
\label{sec:Nisan}
In this section we prove \Cref{thm:constructive-Nisan}, our constructivization of~\hyperlink{nisan-anchor}{Nisan's Theorem}.   We accomplish this using our instance-optimal framework,~\Cref{thm:framework-formal}.  An immediate qualitative difference between~\Cref{thm:instance-opt} and~\Cref{thm:constructive-Nisan} is one sees is that ``there is no $\eps$" in the statement of~\Cref{thm:constructive-Nisan}.  And yet, when applying the framework of~\Cref{thm:framework-formal}, one has to supply the meta-algorithm $\mathcal{A}_{\mathrm{InstanceOpt},\mathcal{E}}$ with an $\eps$ parameter.   Therefore, in order to apply~\Cref{thm:framework-formal} to constructivize~\hyperlink{nisan-anchor}{Nisan's Theorem} (i.e.~to prove~\Cref{thm:constructive-Nisan}), we first have to compute the ``appropriate value of~$\eps$" (\Cref{lemma:proxy error}).

Consider the error metric $\mathcal{E}_{\BayesError}$ defined as follows: 
\begin{align}
    \label{eq: error nisan}
    \mathcal{E}_{\BayesError}(R,D) \coloneqq \Prx_{\bx,\br}[R(\bx, \br) \neq D(\bx)].
\end{align}

The following lemma shows why this this is a useful error function for the purposes of constructivizing~\hyperlink{nisan-anchor}{Nisan's Theorem}: 
\begin{lemma}
    \label{lemma:proxy error}
    For every $\RDT$ $R: \zo^n \times \zo^m \rightarrow \zo$ computing a function $f:\zo^n \to \zo$ with bounded error, there is a unique $\varepsilon_R \in [0,1]$ with the following property. For any $\DDT$ $D:  \zo^n\rightarrow \zo$, if $D \equiv f$\footnote{Meaning that $D(x) = f(x)$ for all $x \in \zo^n$.} then $\mathcal{E}_\BayesError(R,D) = \varepsilon_R$, and $\mathcal{E}_\BayesError(R,D) > \varepsilon_R$ otherwise.
\end{lemma}

\begin{proof}
    Since $R$ computes $f$ with bounded error, we have that for all $x \in \zo^n$, 
    \begin{align*}
        \underbrace{\Prx[R(x, \br) = f(x)] - \Prx[(R(x, \br) \neq f(x)]}_{\coloneqq \Delta(x)} \geq \lfrac{2}{3} - \lfrac{1}{3} = \lfrac{1}{3}.
    \end{align*}
    Denote the quantity on the left side of the above equation as $\Delta(x)$, which is always at least $\frac{1}{3}$. For any $D$, we can write $\mathcal{E}_\BayesError(R,D)$ as follows: 
    \begin{align*}
        \mathcal{E}_\BayesError(R,D) = \Prx[R(\bx, \br) \neq f(\bx)] + \Ex\big[\Ind(D(\bx) \neq f(\bx)) \cdot \Delta(\bx)\big].
    \end{align*}
    Define $\varepsilon_R \coloneqq \mathcal{E}_\BayesError(R, f)$, which is the first term in the above equation. Clearly, if $D\equiv f$, then $\mathcal{E}(R,D) = \varepsilon_R$. Otherwise, since $\Delta(x) > 0$ for all $x$, $\mathcal{E}_\BayesError(R,D) > \varepsilon_R$.
\end{proof}

(Note that $\eps_R$ is precisely the Bayes optimal error of $R$, with $f$ being its Bayes classifier.)  By~\Cref{lemma:proxy error}, if we can find a $\DDT$ $D: \zo^n \to \zo$ minimizing $\mathcal{E}_{\BayesError}(R,D)$ over all $\DDT$s, then $\mathcal{E}_\BayesError(R,D) = \varepsilon_R$ and therefore $D \equiv f$, accomplishing our goal. To apply our instance-optimal framework,~\Cref{thm:framework}, to this error metric $\mathcal{E}_\BayesError$, we need to show that it is natural and efficient (recall~\Cref{def:natural-efficient}): 


\begin{lemma}[$\mathcal{E}_\BayesError$ is natural and efficient]
\label{lem:BO} 
$\mathcal{E}_\BayesError$ is natural and $2^q$-efficient.
\end{lemma}

\begin{proof}
$\mathcal{E}_\BayesError$ is natural since $d(x,y) = |x - y|$ satisfies (\ref{eq:distance}) for $\zo$-valued $R$ and $D$.  (Recall our discussion in Extension \#1 of~\Cref{sec:extensions}.) 
    
    We next show how to efficiently compute $\mathcal{E}_\BayesError(R,D)$. Let the leaves of $D$ be $\ell_1,\ldots \ell_m$ and $\pi_i$ and $\mathrm{label}(\ell_i)$ be defined as follows:
    \begin{align*}
        \pi_i &\coloneqq \text{Path from the root of $D$ to $\ell_i$} \\
        \mathrm{label}(\ell_i) &\coloneqq \text{Leaf value of $\ell_i$.}
    \end{align*}
We can express $\mathcal{E}_\BayesError(R,D)$  as follows:
    \begin{align*}
        \mathcal{E}_\BayesError(R,D) = \sum_{i=1}^m \Prx_{\bx}\big[\bx \text{ follows } \pi_i\big] \cdot 
        \Prx_{\substack{\bx \text{ follows } \pi_i\\ \br \sim \zo^m}} \big[R(\bx, \br) \neq \mathrm{label}(\ell_i) \big]. 
    \end{align*}
    We will show that each of the above terms can be computed efficiently and deterministically. The first term, the probability that $\bx$ follows $\pi_i$ is just $2^{-|\pi_i|}$ where $|\pi_i|$ is the depth of $\ell_i$ in $D$. The second term can be computed using the following relation, which holds since $D$ and $R$ are both $\zo$-valued: 
    \begin{align*}
       \Prx_{\substack{\bx \text{ follows } \pi_i\\ \br \sim \zo^m}} \big[R(\bx, \br) \neq \mathrm{label}(\ell_i) \big] = \Big|\Ex_{\bx \text{ follows } \pi_i}\big[\mu_R(\bx)\big] - \mathrm{label}(\ell_i) \Big|.
    \end{align*}
    The above can be computed efficiently and deterministically by first converting $R$ to $R_{\pi_i}$ and then computing its mean as in the proof of \Cref{lemma:min l2 distance}. Combining each of these steps, we see that Criteria 1 of $2^q$-efficiency (in~\Cref{def:natural-efficient}) is met.  As for Criteria 2, we observe that the constant $c$ minimizing $\mathcal{E}_\BayesError(R,c)$ must either be the constant $0$ or constant $1$ function. We can simply compute the error for both and take whichever is better.
\end{proof}

With~\Cref{lem:BO} in hand, we are now ready to apply our framework,~\Cref{thm:framework-formal}, to give an instance-optimal constructivization of~\hyperlink{nisan-anchor}{Nisan's Theorem}.  

\nisan*

\begin{proof}
    \hyperlink{nisan-anchor}{Nisan's Theorem} guarantees the existence of a $O(q^3)$-query $\DDT$ $D$ that computes $f$ exactly. By \Cref{lemma:proxy error}, we have that $\mathcal{E}_\BayesError(R,D) = \varepsilon_R$, and furthermore this the minimum possible error achievable by any $\DDT$. Therefore, by running $\textsc{Find}(R, \mathcal{E}_\BayesError, \qbudget = O(q^3), \pi = \emptyset)$ we can find a $\DDT$ that achieves error $\varepsilon_R$. Running $\textsc{Find}$ and computing the error of the resulting tree takes time $\mathrm{poly}(N) \cdot n^{O(q^3)}$, at which point our algorithm ``knows" $\varepsilon_R$.  Therefore, we can then use the algorithm of~\Cref{thm:framework-formal}, to find the minimum query $\DDT$ with error $\varepsilon_R$ relative to the error metric $\mathcal{E}_\BayesError$. This step takes time $\mathrm{poly}(N) \cdot 2^{O(q)} \cdot n^{O(q^\star_R)} \le \poly(N)\cdot n^{O(q^3)}$ and returns a $q^\star_R$-query $\DDT$ $D^\star$ with error $\varepsilon_R$ relative to $\mathcal{E}_\BayesError$. By \Cref{lemma:proxy error}, we have that $D^\star$ computes $f$ exactly.
\end{proof}

\begin{remark}
We remark that our \textsc{Find} algorithm (\Cref{fig:find}) as initialized in the proof of~\Cref{thm:constructive-Nisan} can be viewed as a generalization of an algorithm by Mehta and Raghavan \cite{MR02}.   The algorithm of~\cite{MR02} allows one to find a minimal error $q$-query $\DDT$ for a given $\DDT$, where error is measured with respect to Hamming distance.  Our {\sc Find} algorithm initialized with the error metric being $\mathcal{E}_\BayesError$ can be viewed as a generalization of~\cite{MR02}'s algorithm from $\DDT$s to $\RDT$s; indeed, the Bayes error as captured by $\mathcal{E}_\BayesError$ is a natural analogue of Hamming distance for randomized functions.   Without our instance-optimal framework,~\Cref{thm:framework-formal},~\cite{MR02}'s algorithm could also be combined with \Cref{lemma:proxy error} can also be used to constructivize~\hyperlink{nisan-anchor}{Nisan's Theorem}, though not achieving instance optimality.
\end{remark}


\subsection{Consequences of Turing machine computation: Proofs of \Cref{cor:derandomization-with-preprocessing,cor:dequantization-with-preprocessing}}

Our constructivization of~\hyperlink{nisan-anchor}{Nisan's Theorem} (\Cref{thm:constructive-Nisan}) has direct implications for derandomization in the Turing machine model of computation: 

\derandomPreprocess*
\begin{proof}
    Let $\mathcal{A}$ be the randomized $\polylog(n)$-time Turing machine computing $L$.  Note that $\mathcal{A}$ queries at most $\polylog(n)$ coordinates of the input and has randomness complexity at most $\polylog(n)$.  Our preprocessing step first writes down an $\RDT$ $R : \zo^{\polylog(n)} \times \zo^{\polylog(n)} \to \zo$ simulating $\mathcal{A}$, which has size $2^{\polylog(n)} = \quasipoly(n)$, in time $\quasipoly(n)$.  We then apply the algorithm of \Cref{thm:constructive-Nisan} to produce $\polylog(n)$-query $\DDT$ $D$ computing the same function $f : \zo^n \to \zo$ as $R$.  By the guarantees of~\Cref{thm:constructive-Nisan}, doing so also takes time $\quasipoly(n)$.   With this $\polylog(n)$-query DDT $D$ in hand, we can then compute $f(x)$ for any input $x$ in time $\polylog(n)$.
\end{proof}

\dequantizePreprocess*
\begin{proof}
    \cite{BBCMdW01} prove that for any quantum algorithm that makes at most $q$ queries to the input, there is a polynomial of degree at most $2q$ computing the acceptance probability of any $x$. Given a quantum algorithm, their proof implies a method for recovering this polynomial in 
    \begin{align*}
        \poly(2^m, \text{number of terms in the polynomial})
    \end{align*}
    time. Since the polynomial must have degree at most $2q$, that algorithm runs in $\poly(2^m, n^q)$ time. Since we aim to dequantize a quantum algorithm that runs in time at most $\polylog(n)$, it can make at most $\polylog(n)$ queries to the input, so in time $\poly(2^m, \quasipoly(n))$, we can recover a polynomial computing its acceptance probability.
    
    \cite{BBCMdW01} also guarantee that there is an $O(q^6)$ $\DDT$ computing the same Boolean function as a $q$-query quantum algorithm with bounded error. This means there is is $\polylog(n)$-query $\DDT$ deciding $L$ for any particular $n$. Given $p$, a polynomial computing the acceptance probability of the quantum algorithm, we find this $\DDT$ using \Cref{thm:constructive-Nisan} with the following minor modifications. In that proof, we used the following error metric.
    \begin{align*}
    \mathcal{E}_{\BayesOpt}(R,D) \coloneqq \Prx_{\bx,\br}[R(\bx, \br) \neq D(\bx)].
    \end{align*}
    Here, we instead use an error metric that takes in a polynomial and $\DDT$ (as suggested in~\Cref{sec:extensions}, Extension \#2), defined as follows:
    \begin{align*}
        \mathcal{E}_{\mathrm{poly}}(p, D) \coloneqq \Ex[|p(\bx) - D(\bx)|].
    \end{align*}
    These two error metrics would be equivalent if $p$ were a polynomial computing the acceptance probability of $R$, so the proof goes through. Furthermore, when computing what we called $\varepsilon_R$ in \Cref{thm:constructive-Nisan}, we set the query budget to $\qbudget = O(q^6)$ instead of $\qbudget = O(q^3)$. This change affects the time our algorithm takes, but it still runs in the time bounds specified by this lemma. 
    
    The output of the preprocessing is a $\DDT$ that allows us to compute $L$ in $\polylog(n)$ time on any input $x$ of length $n$.  
\end{proof}

\section*{Acknowledgments} 
 
We thank Mika G\"o\"os, Charlotte Peale, and Omer Reingold for enjoyable discussions and helpful suggestions.  The third author is supported by NSF grant CCF-1921795.

\bibliography{most-influential}{}
\bibliographystyle{alpha}


\appendix

\section{Perspectives from learning theory} 
\label{ap:learn}

In this section we briefly discuss a couple of alternative interpretations of the problem of constructively derandomizing query algorithms.  These perspectives come from learning theory, where we adopt the equivalent view of query algorithms as {\sl decision trees} (\Cref{def:DT}).  

Decision trees are an extremely popular model for representing labelled data.  They pervade both the theory and practice of machine learning---their simple structure makes them easy to interpret and fast to evaluate, and they generalize well.  A {\sl  random forest} is a collection of decision trees: to determine the label for an input $x$, the forest simply averages the labels of its trees' labels for $x$.  In other words, if we represent a collection of trees $T_1,\ldots,T_M : \zo^n \to [0,1]$ as $R : \zo^n \times [M] \to [0,1]$ where $R(x,r) = T_r(x)$, then a random forest $\mathscr{F} : \zo^n \to [0,1]$ is the function: 
\[ \mathscr{F}(x) \coloneqq \Ex_{\br \sim [M]}[R(x,\br)] = \Ex_{\br \sim [M]}[T_{\br}(x)] = \mu_R(x). \] 
The motivation for using a collection of trees instead of a single one, supported by the empirical success of random forest algorithms and classifiers, is that its diversity enhances accuracy and stability. 

From this perspective, the task of derandomizing query algorithms corresponds to that of converting a random forest $\mathscr{F}$ into a {\sl single} decision tree that closely approximates $\mathscr{F}$.\footnote{From this perspective---where randomized forests and decision trees are viewed as classifiers rather than a model of computation---it is less common and less natural to make assumptions about the distribution of $R(x,\br)$ (e.g.~that it is concentrated on a certain value), and so the first strand of our results as discussed on~\cpageref{two-strands} is more relevant.}  If one were to do so, one naturally seeks a conversion algorithm that (i) runs quickly, and (ii) preserves the efficiency of the original random forest $\mathscr{F}$, meaning that if $\mathscr{F}$ is a collection of depth-$q$ trees, then the resulting single decision tree has depth $q'$ where $q'$ is not much larger than $q$.  These correspond exactly to the two basic criteria for the efficiency of derandomization that we discuss on~\cpageref{two-criteria} and that we focus on in this work. 

Yet another learning-theoretic interpretation of randomized decision trees is as {\sl latent variable models}: one views randomized decision trees $R : \zo^n\times \zo^m \to [0,1]$ is as deterministic decision trees over $n$ observable variables and $m$ latent variables, where the uncertainty concerning the latent variables is modeled as apparent probabilistic behavior: 
   \[ R : \text{($n$ observable variables)} \times \text{($m$ latent variables)} \to [0,1]. \]  
       This interpretation of randomized decision trees as latent variable models dates back to the original work of Kearns and Shapire~\cite{KS94} extending Valiant's PAC model from deterministic to randomized concepts (which they term ``$p$-concepts"); see Section 3.3 of~\cite{KS94} and their subsequent work with Sellie~\cite{KSS94} for a detailed discussion.         With this interpretation in mind, the algorithmic task of derandomizing randomized decision trees can be viewed as that of efficiently converting a latent variable model into one without any latent variables, while preserving its accuracy as a representation of the data set.
       

\section{Proofs deferred from~\Cref{sec:Yao}}
\label{ap:yao} 

(In this section it will be convenient for us to use notation and terminology, such as ``$\RDT$", ``$\DDT$", and ``$\mu_R$", that we introduce in the Preliminaries section,~\Cref{sec:prelim}.)

\begin{proof}[Proof of~\Cref{fact:yao}] 
    Suppose we pick random strings $\br_1,\ldots, \br_c \sim \zo^{m}$ independently and uniformly at random.  For each $x \in \zo^n$, consider the following random variable:
    \begin{align*}
        \mathbf{est}(x) \coloneqq \Ex_{\bs \in \{\br_1, \ldots, \br_c\}}[R(x, \bs)].
    \end{align*}
    Note that 
    \begin{align*} \E[\mathbf{est}(x)] &= \mu_R(x) = \Ex_{\br \sim\zo^m}[R(x,\br)] \\
    \Var[\mathbf{est}(x)] &= \lfrac{1}{c}\cdot \Varx_{\br\sim\zo^m}[R(x,\br)], 
    \end{align*} 
    where in both cases above, $\br \sim \zo^m$ on the RHS denotes $\br$ chosen uniformly at random from $\zo^m$. 
    Since $R$ has output on the range $[0,1]$, it has variance at most $\frac{1}{4}$. Hence, the variance of $\mathbf{est}(x)$ is at most $\frac{1}{4c}$. If we take $c = \frac{1}{\varepsilon}$, the following holds for any $x \in \zo^n$: 
    \begin{align*}
        \Ex_{\br_1, \ldots, \br_c \sim \zo^m}\Big[\big(\mathbf{est}(x) - \mu_R(x)\big)^2\Big] \leq \frac{\varepsilon}{4}.
    \end{align*}
Next, averaging over $\bx \sim \zo^n$ and swapping expectations, we get: 
    \begin{align*}
        \Ex_{\br_1, \ldots, \br_c \sim \zo^m}\bigg[\Ex_{\bx \sim \zo^n}\Big[\big(\mathbf{est}(\bx) - \mu_R(\bx)\big)^2\Big] \bigg] \leq \frac{\varepsilon}{4}.
    \end{align*}
Therefore, there must exist outcomes $r_1^{\star}, \ldots, r_c^{\star} \in \zo^m$ of $\br_1,\ldots,\br_c$ such that 
    \begin{align}
    \label{eq: existence small error}
        \Ex_{\bx \sim \zo^n}\bigg[\bigg(\Ex_{\bs \sim \{r_1^\star, \ldots, r_c^\star\}}[R(\bx, \bs)] - \mu_R(\bx)\bigg)^2\bigg] \leq \frac{\varepsilon}{4}.
    \end{align}
    For each $i \in [c]$, we consider the $q$-query $\DDT$ computing $x \mapsto R(x, r_i)$ by fixing the stochastic nodes of $R$ according to $r_i^\star \in \zo^m$.  Stacking these $c$ many $q$-query $\DDT$s on top of one another, we have a $\DDT$ that computes $x \mapsto \Ex_{\bs \sim \{r_1^\star, \ldots, r_c^\star\}}[R(x, \bs)]$, which by \Cref{eq: existence small error}, has sufficiently small error.  Since this $\DDT$ makes $q \cdot c = O(q/\eps)$ queries, the proof of~\Cref{fact:yao} is complete.
\end{proof}

\begin{proof}[Proof of~\Cref{fact:lower-bound}]
We first prove the claim for $q=1$.  Consider the $1$-query $\RDT$ $R$ which on input $x$, outputs $x_{\bi}$ where $\bi \sim [n]$ is uniform random.  Let $\eps = \frac1{10n}$ and consider any $q'$-query $\DDT$ $D$. We will show that $\|D - \mu_R\|_2^2 \geq \frac{n-q'}{4n^2}$, which implies that in order for $D$ to $\eps$-approximate $R$, it has to be the case that $q' = \Omega(1/\eps)$.

For $\bx \sim \zo^n$ a uniform random input, the random variable $\mu_R(\bx)$ conditioned on $D$ observing $t$ ones after $q'$ queries is distributed according to 
\[ \frac{\mathrm{Bin}(n-q', \frac1{2})}{n} + \frac{t}{n}.\]  
The variance of this distribution is $\frac{n-q'}{4n^2}$.  Since this lower bounds the approximation error of $D$ with respect to $R$, we have the desired result.

As for $q > 1$, consider the generalization of our construction where we partition the $n$ coordinates into blocks of size $q$. Our $\RDT$ $R$ algorithm picks one of these blocks $\bi \in [\frac{n}{q}]$ uniformly at random and returns the parity of the input coordinates in that block.  An analogous calculation as the one we did for the $q=1$ case above gives the desired lower bound.
\end{proof} 

%
%

\begin{proof}[Proof of~\Cref{fact:pointwise-impossible}]
Let $D$ be any $q$-query that satisfies the pointwise approximation guarantee of~\Cref{fact:pointwise-impossible}, where $q$ is $\le cn$ for some universal constant $c \in (0,1)$ to be chosen later.   As in our proof of~\Cref{fact:lower-bound}, we observe that if $\bx \sim \zo^n$ is a uniform random input, $\mu_R(\bx)$ conditioned on the first $q$ queries of $D$ is distributed according to $\lfrac{\mathrm{Bin}(n-q, 1/2)}{n} + \lfrac{t}{n}$, where $t$ is the number of queries that returned a value of 1. Then there is some  $x^{(1)}$ consistent with the $q$ queries such that $\mu_R(x^{(1)}) = \frac{t}{n}$, and another $x^{(2)}$ consistent with the same queries such that $\mu_R(x^{(2)}) = \lfrac{n-q+t}{n}$.  Consequently, there must also be an $x^{\star}$ consistent with the same queries for which 
\[ |D(x^\star) - \mu_R(x^\star)| \geq \frac1{2} \left(\frac{n-q+t}{n} -  \frac{t}{n}\right) = \frac{n-q}{2n}.\] 
Since $q \le c n$, for large $n$ and for $c$ sufficiently small this difference exceeds $0.1$, which concludes the proof of~\Cref{fact:pointwise-impossible}.
\end{proof}


\end{document}